\documentclass[12pt]{article}
\usepackage{amsmath}
\usepackage{amssymb}
\usepackage{amsthm}
\usepackage{tikz}
\usetikzlibrary{decorations.markings,decorations.pathreplacing,calc}
\usetikzlibrary{arrows.meta}
\usepackage{algorithm}
\usepackage{algpseudocode}
\usepackage{thm-restate}
\usepackage{verbatim}
\usepackage[framemethod=tikz]{mdframed}
\usepackage{complexity}

\usepackage[margin=1in]{geometry}
\usepackage{subfig}
\usepackage[utf8]{inputenc}

\usepackage[colorlinks = true]{hyperref}

\usepackage{xcolor}
\definecolor{darkred}  {rgb}{0.5,0,0}
\definecolor{darkblue} {rgb}{0,0,0.5}
\definecolor{darkgreen}{rgb}{0,0.5,0}

\hypersetup{
  urlcolor   = blue,         
  linkcolor  = darkblue,     
  citecolor  = darkgreen,    
  filecolor  = darkred       
}

\newcommand{\be}{\begin{equation}}
\newcommand{\ee}{\end{equation}}
\newcommand{\ba}{\begin{array}}
\newcommand{\ea}{\end{array}}
\newcommand{\bea}{\begin{eqnarray}}
\newcommand{\eea}{\end{eqnarray}}

\newcommand{\calH}{{\cal H }}

\newcommand{\la}{\langle}
\newcommand{\ra}{\rangle}
\newcommand{\tr}{\text{Tr}}

\renewcommand{\EE}{\mathbb{E}}

\newtheorem{dfn}{Definition}

\newtheorem{claim}{Claim}
\newtheorem{lemma}{Lemma}

\newtheorem{fact}{Fact}

\newtheorem{theorem}{Theorem}
\newtheorem*{theorem*}{Theorem}

\newtheorem{subtheorem}{Theorem}[theorem]

\newcommand{\footremember}[2]{%
    \footnote{#2}
    \newcounter{#1}
    \setcounter{#1}{\value{footnote}}%
}
\newcommand{\footrecall}[1]{%
    \footnotemark[\value{#1}]%
}
\title{Classical simulation of peaked shallow quantum circuits}
\author{Sergey Bravyi\footremember{ibm}{IBM Quantum, IBM T.J. Watson Research Center}
\and
David Gosset\footremember{iqc}{Institute for Quantum Computing, University of Waterloo, Canada}\footremember{co}{Department of Combinatorics and Optimization, University of Waterloo, Canada}\footremember{PI}{Perimeter Institute for Theoretical Physics, Waterloo, Canada}
\and Yinchen Liu\footrecall{iqc} \footrecall{co} \footrecall{PI}
}
\date{}
\begin{document}
\maketitle
\begin{abstract}

An $n$-qubit quantum circuit is said to be \textit{peaked} if it has an output probability that is at least inverse-polynomially large as a function of $n$. 
We describe a classical algorithm with quasipolynomial runtime $n^{O(\log{n})}$ that approximately samples from the output distribution of a peaked constant-depth circuit. We give even faster algorithms for circuits composed of nearest-neighbor gates on a $D$-dimensional grid of qubits, with polynomial runtime $n^{O(1)}$ if $D=2$ and almost-polynomial runtime $n^{O(\log{\log{n}})}$ for $D>2$.  Our sampling algorithms can be used to estimate output probabilities of shallow circuits to within a given inverse-polynomial additive error, improving previously known methods. As a simple application, we obtain a quasipolynomial algorithm to estimate the magnitude of the expected value of any Pauli observable in the output state of a shallow circuit (which may or may not be peaked). This is a dramatic improvement over the prior state-of-the-art algorithm which had an exponential scaling in $\sqrt{n}$.

\end{abstract}
\section{Introduction}

Whether a real-world quantum computer outperforms classical computers is a delicate question. Can we solve an instance of a problem or perform a task for which there is some \textit{asymptotic} evidence of quantum speedup? Can we draw instances from a classically hard distribution? Is the problem size large enough that existing classical computers cannot solve it? Can we verify that the output of the quantum computer is correct?  The situation is further complicated by the fact that near-term quantum computers are highly constrained:  their architecture or gate-set may be limited in some way, and they are noisy, so performance degrades as the circuit size grows.

But even if we focus on idealized, simple, restricted models of computation, we are faced with navigating a suprisingly treacherous boundary between quantum and classical computing. For instance, consider the case of shallow (i.e., constant-depth) quantum circuits with nearest-neighbor gates in a two-dimensional architecture. In this setting there is good complexity-theoretic evidence for a certain kind of quantum advantage: sampling from the output distribution of shallow 2D circuits, even approximately, is expected to be intractable for classical machines in the worst case \cite{bermejo2018architectures}.  But the extent to which this is generic is unclear: shallow 2D circuits where the gates are chosen at random admit an empirically successful classical simulation algorithm based on tensor network contraction, which is conjectured to be efficient for circuit depths below a critical value \cite{napp2022efficient}. Moreover, unfortunately, there is no known efficient method for verifying that a collection of samples have been drawn correctly from the output distribution. And this sampling task is somewhat unsatisfying---what problem does it solve? As a more practical alternative, one might instead try to implement a variational quantum algorithm, which have been proposed as a flexible algorithmic framework for solving real-world quantum chemistry problems and classical optimization problems \cite{peruzzo2014variational,kandala2017hardware}. In these algorithms the job of the quantum computer is to repeatedly estimate expected values of simple (tensor product) observables at the output of a quantum circuit.  Unfortunately, this task is classically easy for 2D shallow circuits: mean values can be approximated efficiently on a classical computer \cite{bravyi2021classical}. Finally, we note that 2D shallow circuits can solve a linear algebra problem that cannot be solved by their classical counterparts, or even the slightly broader family of general constant-depth classical circuits \cite{bravyi2018quantum}. However, this is not a dramatic quantum speedup---the problem can be solved in polynomial time by classical computers.

This example illustrates the rich and varied complexity-theoretic landscape that emerges from a simple theoretical model of small quantum computers. To navigate this landscape---e.g., to identify promising candidates for quantum speedups---it is necessary to refine both our understanding of quantum algorithms as well as classical simulability of quantum computers. Towards this end, here we focus on shallow quantum circuits and describe new classical simulation algorithms that improve upon prior work in several ways.

We are mainly interested in classical simulation algorithms capable of sampling the output distribution 
or approximating the entire output state of a quantum circuit 
with a small (inverse polynomial in $n$) error.
However, to set the stage for our results, let us first discuss a simpler task ---  \textit{output probability estimation}.
That is, given an $n$-qubit constant-depth quantum circuit $U$, a desired precision $\epsilon=\Omega{(n^{-c})}$  (where $c$ is a constant) and a binary string $x\in \{0,1\}^n$, we aim to output an estimate $p$ satisfying 
\begin{equation}
|p-|\langle x|U|0^n\rangle|^2|\leq \epsilon.  
\label{eq:invpolyerror}
\end{equation} 
Ref. \cite{bravyi2021classical} showed that this problem admits a polynomial-time classical algorithm if the gates of $U$ act between nearest-neighbors in a two-dimensional array of qubits. This was extended to a quasipolynomial algorithm for three-dimensional circuits in Ref.\cite{coudron} and for $D$-dimensional circuits for $D=O(1)$ in Refs.~\cite{coudron, coudron2}. For general shallow circuits lacking geometric locality the best previously known algorithm had an exponential scaling in $\sqrt{n}$ \cite{bravyi2021classical}.
As a simple consequence of our main technical results (described below), we obtain improved classical algorithms
for the above problem summarized in Table \ref{table:probabilities}. 
Notably, we obtain a quasipolynomial classical algorithm for output probability estimation for general shallow quantum circuits, an exponential improvement in runtime over prior work. We also improve the asymptotic scaling of the best prior algorithm for $D$-dimensional shallow circuits from Refs.~\cite{coudron, coudron2}. For two-dimensional circuits we present a new algorithm that, like prior work \cite{bravyi2021classical}, has a polynomial scaling.  
However, as we now explain, the algorithms presented here go beyond prior work in that they are capable of solving a more general problem 
than output probability estimation.

\renewcommand{\arraystretch}{2} 

\begin{figure}
\begin{center}
\begin{tabular}{ |c|c|c| } 
\hline
& Prior work & Our results\\
 \hline
$2$-dimensional  & $\mathrm{poly}(n)$\cite{bravyi2021classical} & $\mathrm{poly}(n)$ \\ 
\hline
 $D$-dimensional, $D=O(1)$ & $n^{O(\mathrm{poly}(\log{n}))}$ \cite{coudron, coudron2} & $n^{O(\log(\log{n}))}$   \\ 
\hline
 General & $e^{O(\sqrt{n\log{n}})}$\cite{bravyi2021classical} &   $n^{O(\log{n})}$ \\ 
 \hline
\end{tabular}
\caption{Runtime of classical algorithms for approximating output probabilities of $n$-qubit shallow quantum circuits. \label{table:probabilities}In contrast with prior work, the algorithms presented here solve the more general problem of sampling from 
the output distribution of peaked shallow circuits to within a given error inverse polynomial in $n$.}
\end{center}
\end{figure}

To describe our results in more detail, note that an estimate of the form Eq.~\eqref{eq:invpolyerror} is only informative if 
at least one output probability of the circuit $U$ is inverse polynomially large as a function of the number of qubits (otherwise, $p=0$ would be a suitable estimate for any output $x$). Quantum circuits which have inverse polynomially large output probabilities are said to be \textit{peaked}. \ Aaronson has suggested that peaked quantum circuits might be useful in charting a ``feasible route to near-term quantum supremacy" \cite{aaronson2022much}. In particular, if one runs a peaked quantum circuit on a quantum computer and samples from the output distribution polynomially many times, there is a reasonable chance of obtaining an outcome $x$ on which it is peaked. If we know $x$ beforehand (perhaps due to the way the circuit was generated), then this could be useful as a way to partially verify the quantum computation. 

For a nonnegative integer $a$, let us say that an $n$-qubit unitary $U$ is $a$-peaked if 
\begin{equation}
\max_{x\in \{0,1\}^n} |\langle x|U|0^n\rangle|^2 \geq n^{-a}.
\label{eq:peakeddef}
\end{equation}
 We say a circuit family $\{U_n\}$ is peaked if there is some $a\in \mathbb{Z}_{\geq 0}$ such that $U_n$ is $a$-peaked for all $n\geq 0$.

Our first result is a quasipolynomial simulation algorithm for peaked shallow circuits.  Let $c>0$ be an arbitrary constant.
\begin{theorem}[\textbf{Output distribution of peaked shallow circuits}]
There exists a classical algorithm which takes as input a positive integer $a=O(1)$, an $n$-qubit constant-depth circuit $U$, and a precision parameter $\epsilon=\Omega(n^{-c})$. If the algorithm succeeds then it outputs a sample from a probability distribution $P'$ over $n$-bit strings such that $\|P'-P\|_1\leq \epsilon$ where $P(x)=|\langle x|U|0^n\rangle|^2$ is the output distribution of the circuit; otherwise it outputs an error flag. The algorithm is guaranteed to succeed if $U$ is $a$-peaked. The runtime of the algorithm is $n^{O(\log{n})}$.
\label{thm:main1}
\end{theorem}

 If we specialize to the problem of
estimating output probabilities of $D$-dimensional shallow circuits, the algorithm described in Theorem \ref{thm:main1} gives an improved runtime scaling compared with the
previous state-of-the-art~\cite{coudron, coudron2} algorithms which were tailored to that special case, see Figure~\ref{table:probabilities}.
We emphasize that our algorithm solves a more general problem and is much simpler than the ones from Refs.~\cite{coudron, coudron2}. It is sufficiently simple that we were able to implement it in software and simulate peaked shallow circuits with about 50 qubits, see Section~\ref{sec:numerics} for details.

To understand the proof of Theorem \ref{thm:main1}, first consider the very special case in which the output distribution of $U$ is a product distribution $P(x)=p_1(x_1)p_2(x_2)\ldots p_n(x_n)$, where each $p_j$ is a distribution over a single bit $x_j\in \{0,1\}$. Let us suppose $P$ is of this form and in addition Eq.~\eqref{eq:peakeddef} holds. For simplicity let us assume the probability of obtaining the all-zeros string is inverse-polynomially large, i.e., $|\langle 0^n|U|0^n\rangle|^2\geq n^{-a}$. Then the expected number of zeros in a string $x$ sampled from $P$ is (using the arithmetic-geometric mean inequality)
\[
\sum_{i=1}^{n} p_i(0)\geq n\left(p_1(0)p_2(0)\ldots p_n(0)\right)^{\frac{1}{n}}
=nP(0^n)^{\frac1n}
\geq n^{1-a/n} =n-O\left(\log{n}\right).
\]
That is, we expect $x$ to have $O(\log{n})$ ones.  This argument can be further strengthened using a Chernoff bound, to show that the support of $P$ is concentrated on binary strings of Hamming weight at most $O(\log{n})$. The classical simulation algorithm is then based on restricting our attention to the subspace of Hilbert space corresponding to these $n^{O(\log{n})}$ computational basis vectors. For general shallow circuits,  the bits in the output distribution are not independent random variables. But dependency is highly constrained: each bit is independent of all but a constant number of other bits. In the proof of Theorem \ref{thm:main1} we show that the above concentration of measure argument can be extended to the setting of shallow quantum circuits. The proof uses a large deviation bound for partially dependent random variables due to Janson \cite{janson2004large}. In this way we also establish that the output distribution of a peaked shallow circuit is concentrated in a Hamming ball of diameter $O(\log{n})$. In contrast, for shallow quantum circuits that may not be peaked, the Hamming weight is only guaranteed to be concentrated in an interval of width $O(\sqrt{n})$ around its mean value~\cite{anshu2022concentration} \footnote{It is again instructive to consider the special case of product distributions, which illustrates optimality in this setting as well.}. Our simulation algorithm  combines the Hamming weight concentration result with a simple circuit-to-Hamiltonian mapping.
By projecting a local Hamiltonian $\sum_{j=1}^n U|0\ra\la 0|_j U^\dag$ associated with the circuit $U$ onto the Hamming ball of diameter $O(\log{n})$ and computing an
eigenvector of the projected Hamiltonian with the largest eigenvalue we obtain a good approximation of the output state $U|0^n\ra$ which is specified
by its $n^{O(\log{n})}$ nonzero amplitudes in the computational basis. 
If the largest eigenvalue of the projected Hamiltonian is below a certain threshold value, the algorithm infers that the circuit $U$ is not $a$-peaked and outputs an error flag.
This yields the following version of Theorem~\ref{thm:main1}.

\begin{subtheorem} [\textbf{Output state of peaked shallow circuits}]
There exists a classical algorithm which takes as input a positive integer $a=O(1)$, an $n$-qubit constant-depth circuit $U$, and a precision parameter $\epsilon=\Omega(n^{-c})$. If the algorithm succeeds then it outputs a classical description of an $n$-qubit state $|\phi\rangle$ specified by its $n^{O(\log{n})}$ nonzero amplitudes in the computational basis, and such that $|\langle \phi|U|0^n\rangle|^2\geq 1-\epsilon$; otherwise it outputs an error flag. The algorithm is guaranteed to succeed if $U$ is $a$-peaked. The runtime of the algorithm is $n^{O(\log{n})}$.
\label{thm:mainstate}
\end{subtheorem}

Theorem \ref{thm:main1} is a direct corollary of Theorem \ref{thm:mainstate}, obtained by using the relationship between the fidelity and trace distance, and the fact that trace distance does not increase under CPTP maps (in this case, measuring all qubits in the computational basis).

Theorems \ref{thm:main1}, \ref{thm:mainstate} apply to general shallow quantum circuits in which two-qubit gates may be applied to any pair of qubits. We also present sampling algorithms with improved asymptotic runtime for shallow circuits composed of two-qubit gates acting between nearest neighbors on a grid in a fixed number of dimensions.
\begin{theorem}[\textbf{Output distribution with geometric locality}]
Under the conditions of Theorem \ref{thm:main1}, if in addition $U$ is local with gates that act between nearest neighbors on a two-dimensional grid of qubits then the same task can be achieved by a classical algorithm with $\mathrm{poly}(n)$ runtime. If it is geometrically local in $D$-dimensions, with $D=O(1)$ then there is a classical algorithm with runtime $n^{O(\log\log{n})}$. 
\label{thm:main2}
\end{theorem}

The proof of Theorem \ref{thm:main2} is based on a different technique from that of Theorem \ref{thm:main1}. It uses some ideas from Ref.~\cite{coudron} but deploys them in a novel way. The key technique is a kind of dissection along one side of the grid that we now explain. For concreteness consider the case of a 2D shallow circuit $U$ on a $\sqrt{n}\times \sqrt{n}$ square grid of qubits which satisfies $|\langle 0^n|U|0^n\rangle|^2\geq n^{-a}$.   In the first step of the proof we show how to find $O(\sqrt{n})$ \textit{heavy slices} arranged in a somewhat regular pattern. By definition, each heavy slice is a rectangular strip of qubits of size $O(1)\times \sqrt{n}$ that has a good probability, say at least $0.99$, of being measured to be in the all-zeros state. Furthermore, tracing out any heavy slice leaves the two remaining regions on the left and on the right  in a tensor product state. These heavy slices are arranged in a regular pattern so that the maximum distance between heavy slices is at most $O(\log{n})$. 
Imagine taking the true output state of the circuit $|\psi\rangle=U|0^n\rangle$ and swapping all the heavy slices with fresh qubits prepared in the all-zeros state. The result is a mixed state
\begin{equation}
\sigma'= \sigma^{1}\otimes \sigma^{2}\otimes \ldots \otimes \sigma^{T}
\label{eq:tensormps}
\end{equation}
where each $\sigma^{i}$ coincides with a reduced density matrix of the output state $|\psi\rangle$ on some quasi-one dimensional rectangular strip of qubits of size $O(\log{n})\times \sqrt{n}$ along with some additional qubits in the $|0\rangle$ state. The output distribution of $\sigma'$ can be efficiently sampled using matrix product state techniques. Unfortunately, the measurement statistics of $\sigma'$ do not in general approximate those of $\psi$ because our heavy slices only come with a weak guarantee ($0.99$ probability of measuring all zeros). We show that this kind of weak approximation can be improved using an inclusion-exclusion method; the result is an approximation $\rho$ to $\psi$ that is  a \textit{real linear combination}, or ``pseudomixture",  of simple quantum states of the form Eq.~\eqref{eq:tensormps}. Although the linear combination contains exponentially many terms,  its structure is such that marginals like $\mathrm{Tr}(\rho |x\rangle\langle x|_S)$ for $S\subseteq [n]$ and $x\in \{0,1\}^{|S|}$, can be expressed as a matrix-vector product and computed efficiently \footnote{Matrix multiplication is used in two distinct ways here---an outer matrix multiplication is used to reduce the task of computing expected values of $\rho$ to computing expected values of matrix product states, and an inner matrix multiplication is used to perform the latter task.}.  To approximately sample from the output distribution of $\psi$, we use a sampling-to-marginal reduction from Ref.~\cite{sparse} that is robust to error in the $1$-norm. For $D$-dimensional circuits with $D>2$ a similar dissection technique is used inductively.

Theorems \ref{thm:main1} and \ref{thm:main2} show that the task of approximately sampling from the output distribution---which is believed to be classically hard even for 2D shallow circuits---becomes easy when the circuit is both shallow and peaked. On the other hand, polynomial-sized peaked quantum circuits (which may have a high circuit depth) are as powerful as a universal quantum computer and very unlikely to admit an efficient or quasipolynomial classical simulation.

The remainder of paper is structured as follows. Below, in Section \ref{sec:applications}, we describe some applications of our algorithms including 
the estimation of output probabilities, quantum mean values, the normalized trace of unitaries and related distance measures, and more. In Section \ref{sec:preliminaries} we review some basic terminology and properties of shallow quantum circuits. In Section \ref{sec:ham} we describe the Hamming weight concentration properties of peaked shallow circuits and give the proof of Theorems \ref{thm:main1},\ref{thm:mainstate}. In Section \ref{sec:numerics} we present data from a software implementation of this algorithm which illustrates the range of its utility in practice. In Section \ref{sec:kdim} we consider geometrically local shallow circuits and establish Theorem \ref{thm:main2}. A summary of our results and some open questions can be found in Section~\ref{sec:conclusions}. 

\subsection{Applications}
\label{sec:applications}
Our algorithms for simulating peaked shallow circuits have several applications for simulation of general (not necessarily peaked) shallow circuits.
\\

\noindent \textit{Output probabilities, and magnitude of quantum mean values}
We obtain randomized classical algorithms for output probability estimation, summarized in Table \ref{table:probabilities}, as a straightforward application of our sampling algorithms for peaked shallow circuits. Indeed, suppose we aim to compute an estimate $p$ satisfying Eq.~\eqref{eq:invpolyerror} for some given shallow circuit $U$,  $\epsilon\geq n^{-c}$, and binary string $x\in \{0,1\}^n$ (the circuit $U$ may or may not be peaked). 
To do so we can use the algorithm from Theorem \ref{thm:main1} (or Theorem \ref{thm:main2} if $U$ is geometrically local) with $a=c$ and error tolerance $\epsilon/2$. We repeat the sampling algorithm $N$ times. If in any of these runs the algorithm outputs an error flag, we set $p=0$; otherwise,  let $p$ be the fraction of times that outcome $x$ is obtained. In the former case, the circuit is not $c$-peaked, which implies that $p=0$ satisfies Eq.~\eqref{eq:invpolyerror}. In the latter case, letting $q$ be the distribution sampled by the algorithm we have, with high probability
\[
|p-|\langle x|U|0^n\rangle|^2|\leq |p-q(x)|+|q(x)-|\langle x|U|0^n\rangle|^2|\leq \epsilon/2 +O(N^{-1/2}).
\]
Taking $N=O(\epsilon^{-2})=O(\mathrm{poly}(n))$ suffices to make the second term at most $\epsilon/2$.

 A simple application of output probability estimation is to additively approximate the square (or magnitude) of expected values of tensor product observables. Suppose $U$ is a depth $d=O(1)$ circuit. We do not assume that $U$ is peaked.
 Given single-qubit observables $O_1,O_2,\ldots O_n$ such that $\|O_j\|\leq 1$ for all $j\in [n]$, consider the expected value squared:
\begin{equation}
|\langle 0^n|U^{\dagger} O_1\otimes O_2\otimes \ldots \otimes O_n U|0^n\rangle|^2.
\label{eq:meanval}
\end{equation}
If each observable $O_j$ is unitary, then Eq.~\eqref{eq:meanval} is an output probability of a circuit with depth $2d+1$, and therefore it can be estimated to within an inverse polynomial additive error by directly applying our algorithms for output probability estimation. If on the other hand some or all of the observables $O_j$ are not unitary, we can still express the mean value Eq.~\eqref{eq:meanval} as an output probability of a depth $2d+1$ quantum circuit acting on at most $2n$ qubits. This follows from the simple fact that for any single qubit operator $O$ satisfying $\|O\|\leq 1$,  we can always find a two-qubit unitary $B$ such that $(I\otimes\langle 0|) B (I\otimes|0\rangle)=O$, see e.g., Lemma 5 of the Supplementary Material of Ref.\cite{bravyi2021classical}. 

Consequently, we obtain a quasipolynomial algorithm to additively estimate the magnitude of any Pauli operator at the output of a constant-depth quantum circuit. However, we do not know how to use the algorithms presented here (or the ones from \cite{coudron, coudron2}) to compute the overall $\pm$ sign of the Pauli expected value---it remains a challenging open question whether quantum mean values of general shallow circuits (including the sign) can be additively estimated with an efficient or quasipolynomial classical algorithm. For 2D circuits the efficient algorithm from Ref.~\cite{bravyi2021classical} is capable of computing expected values of tensor product observables, including the sign. Intriguingly, Ref.~\cite{huang2021information} describes a learning task for which the sign of Pauli expected values is, in a sense,  provably classically harder than their magnitude. 
\\

\noindent \textit{Normalized trace, Frobenius distance, and frame potentials}

Another interesting application of our algorithms relate to the One Clean Qubit model of a quantum computation, also known as DQC1~\cite{knill1998power}.
This model considers quantum circuits on $n+1$ qubits with the mixed initial state $\rho_{in} = |0\ra\la 0|\otimes I/2^{n}$
such that all qubits except for the first one are maximally mixed. After applying a polynomial size quantum circuit to all $n+1$ qubits
some designated output qubit is measured in the $0,1$ basis. The goal is to estimate the probability of measuring each outcome. 
The canonical DQC1-complete problem is estimation of the normalized trace
\[
t(U)=\frac1{2^n}\mathrm{Tr}(U),
\]
where $U$ is a polynomial-size quantum circuit on $n$ qubits~\cite{knill1998power}. One has to estimate $t(U)$ within additive error
$\mathrm{poly}(1/n)$.  The DQC1 model can efficiently 
estimate $t(U)$ (both real and imaginary parts) using the  Hadamard test.
Here we are interested in the special case when $U$ is a constant-depth circuit. 
Is there an efficient classical algorithm that can estimate $t(U)$ with a small additive error?

Here we shed light on this question by establishing the following.
\begin{claim}
\label{corol:trace}
Suppose $U$ is a depth-$d$ circuit acting on $n$ qubits, and let $t(U)=(1/2^n)\mathrm{Tr}(U)$ be the normalized trace of $U$.
There exists a $2n$-qubit quantum circuit $W$ of depth $d+4$ such that  
\[
|\la 0^{2n}|W|0^{2n}\ra|^2=|t(U)|^2.
\]
\end{claim}
\begin{proof}
The quantum circuit $W$ is shown on Figure~\ref{fig:bell}.
A simple calculation shows that 
\[
\la 0^{2n}|W|0^{2n}\ra = \la \Phi|U\otimes I|\Phi\ra=  t(U),
\]
where $|\Phi\ra=(1/\sqrt{2^n})\sum_{x\in \{0,1\}^n} |x\ra\otimes |x\ra$ is the Bell state of $n+n$ qubits.
Thus the output probability $|\la 0^{2n}|W|0^{2n}\ra|^2$ coincides with $|t(U)|^2$.
\end{proof}

\vspace{1cm}
\begin{figure}[hbt]
\centerline{
\includegraphics[height=5cm]{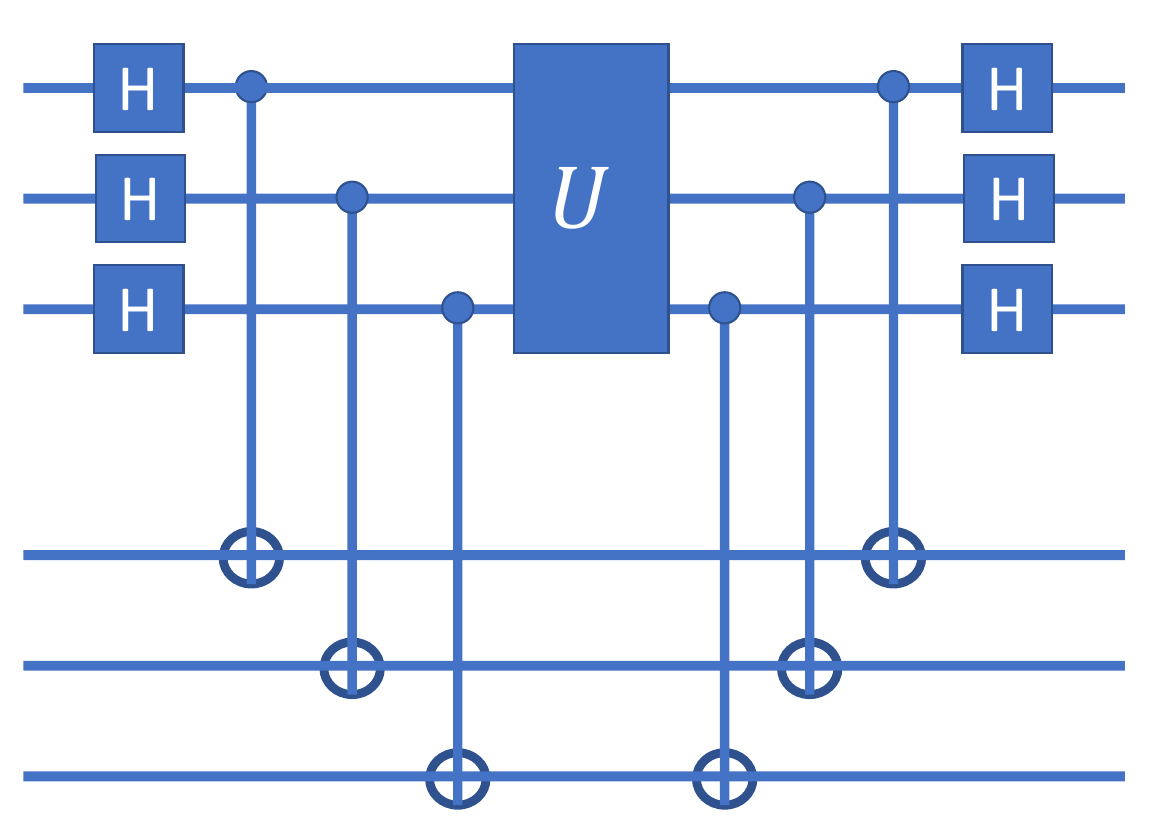}}
  \caption{Circuit $W$ from Claim~\ref{corol:trace}. Here $n=3$.} \label{fig:bell}
\end{figure}
\vspace{1cm}

Claim~\ref{corol:trace} shows that, for constant-depth $U$, the square of the normalized trace $|t(U)|^2$ can be expressed as the output probability of a shallow circuit and estimated using the algorithms described in Table \ref{table:probabilities}. Although this falls short of estimating $t(U)$ including the overall phase, it is enough to check whether two constant-depth circuits $U$ and $V$ are close to the each other in the Frobenius norm.
Since the overall phase of a quantum circuit is often  irrelevant, it is natural to quantify the distance between $U$ and $V$
as
\[
\delta(U,V) :=  \frac1{2^n}
\min_{\gamma \in [0,2\pi)} \| e^{i\gamma} U - V\|_F^2,
\]
where $\|\cdot \|_F$ is the Frobenius norm. A simple calculation shows that 
\[
\delta(U,V) = 2(1-|t(UV^\dag)|).
\]
If both $U$ and $V$ are constant-depth then $UV^\dag$ is also constant-depth.
Thus our algorithm can estimate $\delta(U,V)$ with a small additive error. 
In contrast, checking whether
two constant-depth circuits $U$ and $V$ are close to each other in the operator (spectral) norm
 is known to be a QMA-complete problem,
even in the case of 1D constant-depth circuits~\cite{ji2009non}.

We can also estimate certain distance measures involving random ensembles of unitaries, by a reduction to estimating quantities of the form $|t(U)|^{2k}$ for constant $k$. In particular, suppose $\mathcal{D}$ is a distribution over constant-depth $n$-qubit quantum circuits, such that a circuit can be drawn from $\mathcal{D}$ in $\mathrm{poly}(n)$ time. For any constant $k\geq 2$, consider a normalized version of the $k$-th frame potential of $\mathcal{D}$, which is a measure of how close the distribution is to a unitary $k$-design:
\[
f^{(k)}_\mathcal{D}\equiv 2^{-2nk}\mathbb{E}_{U,V\sim \mathcal{D}} |\mathrm{Tr}(U^{\dagger}V)|^{2k} =\mathbb{E}_{U,V\sim \mathcal{D}} |t(U^{\dagger}V)|^{2k}.
\]
This normalized frame potential satisfies \cite{zhu2017multiqubit}
\[
\frac{k!}{2^{2nk}}\leq f^{(k)}_\mathcal{D}\leq 1
\]
where the lower bound is achieved iff $\mathcal{D}$ is a $k$-design. By sampling circuits from $\mathcal{D}$ and computing the sample mean of $|t(U^{\dagger}V)|^{2k}$ we can estimate $f^{(k)}_\mathcal{D}$ to within a given inverse polynomial additive error.
\\

\noindent \textit{Approximation via linear combination of low-depth states}

Suppose $U$ is a quantum circuit and we are interested in preparing its output state $U|0^n\rangle$ on a quantum computer. We'd be very happy if we could find a much simpler quantum circuit $V$, say of significantly lower depth than $U$, such that the output state of $VU$ is completely concentrated on the all-zeros string, i.e.,   
\begin{equation}
\langle 0^n|V U|0^n\rangle=1.
\label{eq:uv}
\end{equation}
Then obviously $U|0^n\rangle=V^{\dagger}|0^n\rangle $ and thus
$V^{\dagger}$
gives a simpler way to prepare the output state of $U$.

Using Theorem \ref{thm:mainstate} we can make a version of this work even when we have a much weaker guarantee than Eq.~\eqref{eq:uv}.  In particular, suppose that $U$ is a depth $d=O(1)$ quantum circuit that admits a kind of approximate shortcut, in the sense that there is a circuit $V$, of lower depth $d'<d$ such that $VU$ is peaked, i.e., $\langle 0^n|VU|0^n\rangle\geq n^{-a}$ for some $a$. For a small (inverse polynomial) $\delta$ we can then use Theorem \ref{thm:mainstate} to obtain complex coefficients $\{c_x\}$ such that
\[
|\phi\rangle=\sum_{|x|\leq O(\log{n})} c_x |x\rangle  \quad \text{satisfies} \quad |\langle \phi|VU|0^n\rangle|^2\geq 1-\delta
\]
In this way we obtain an approximation to the output state of the depth-$d$ circuit $U$, as a superposition of quasipolynomially many states of (lower) depth $d'$, i.e. $V^{\dagger}|\phi\rangle=\sum_{|x|\leq O(\log{n})} c_x V^{\dagger}|x\rangle$.

\section{Preliminaries}
\label{sec:preliminaries}
In this work we consider quantum circuits composed of two-qubit gates. A depth-$1$ quantum circuit consists of a tensor product 
of two-qubit gates with disjoint support. A depth-$d$ quantum circuit is a sequence of $d$ depth-$1$ quantum circuits.  We are interested in constant-depth, or \textit{shallow},  quantum circuits such that $d$ is a constant independent of the number of qubits $n$. 

\subsection*{Lightcones and independence}
We can associate a directed graph $G$ with an $n$-qubit quantum circuit $\mathcal{C}$ in a straightforward way as follows. Starting from a circuit diagram for $\mathcal{C}$, we add a vertex of $G$ for each qubit at the input of the circuit, a vertex for each gate in the circuit, and a vertex for each qubit at the output of the circuit. Each wire in the circuit is then assigned a directed edge from left-to-right between the vertices at its endpoints. Then the circuit depth can be equivalently defined as $m-1$ where $m$ is the number of edges in the longest path from a vertex representing an input qubit vertex to one representing an output qubit. For any qubit $j\in [n]$ we  also define the \textit{lightcone} $\mathcal{L}(j)$ to be the set $S\subseteq [n]$ of input qubits  that are connected by a directed path in $G$ to the output qubit $j$. Similarly, we define
\[
\mathcal{L}(A)\equiv \bigcup_{j\in A} \mathcal{L}(j) \quad A\subseteq [n].
\]
If two subsets of qubits $A,B\subseteq [n]$ have the property that $\mathcal{L}(A)\cap \mathcal{L}(B)=\emptyset$ then we say $A,B$ are \textit{lightcone-separated}.
\begin{fact}
 Let $\rho=U|0^n\rangle\langle 0^n|U^{\dagger}$ be the output state of the circuit $U$ and suppose $A,B$ are lightcone-separated. Then
\[
\rho_{AB}=\rho_A\otimes \rho_B
\]
where $\rho_S=\mathrm{Tr}_{[n]\setminus S}(\rho)$. 
\label{fact:indep}
\end{fact}
A key property of shallow quantum circuits is that the lightcone of any qubit has a constant size. In particular, $|\mathcal{L}(j)|\leq 2^d$ for all $j\in [n]$. Each qubit $j\in [n]$ is therefore lightcone-separated from all but a constant number $K\leq 2^{2d}$ of the other qubits. Combining this with Fact \ref{fact:indep}, we see that in the output distribution $P(x)=|\langle x|U|0^n\rangle|^2$,  each bit $x_j$ is independent of all but a constant number $K\leq 2^{2d}$ of other bits.

\subsection*{Peaked shallow purification}

  Here we consider a reduced density matrix $\rho$ of an $n$-qubit shallow circuit describing a subset of qubits $A\subseteq [n]$ such that $\langle h|\rho|h\rangle\geq n^{-a}$ for some heavy string $h\in \{0,1\}^{|A|}$. The following lemma shows, in particular,  that $\rho$ admits a purification which is the output state of a peaked shallow circuit. For technical reasons we shall use the slightly more general statement below which describes a family of shallow circuits which each purify $\rho$, and at least one of which is peaked.
\begin{lemma}[Peaked shallow purification]
Suppose $V$ is an $n$-qubit depth-$d$ quantum circuit and suppose the qubits are partitioned as $[n]=AB$. Let $\rho=\mathrm{Tr}_{B}\left(V|0^n\rangle\langle0^n|V^{\dagger}\right)$ be the reduced density matrix of the output state on subsystem $A$. Suppose further that for a given $h\in \{0,1\}^{|A|}$ we have
\[
\mathrm{Tr}\left( \rho|h\rangle\langle h|_A\right)\geq n^{-a}.
\]
For each $z\in \{0,1\}^{|A|}$, let $X(z)$ be an $|A|$-qubit Pauli matrix such that $X(z)|0^{|A|}\rangle=|z\rangle$. Define unitaries $Q(z)$ on three registers $AA'B$, where $|A'|=|A|$, and 
\begin{equation}
Q(z)=(I_A\otimes V_{A'B}^{\dagger})(V_{AB}\otimes X(z)_{A'}).
\label{eq:udef}
\end{equation}
$Q(z)$ acts on $m=2|A|+|B|\leq 2n$ qubits and has depth $2d$. Moreover, we have
\[
\max_{x\in \{0,1\}^m}|\langle x|Q(h)|0^{m}\rangle|^2\geq n^{-2a}
\]
and
\[
 \mathrm{Tr}_{A'B}(Q(z)|0^{m}\rangle\langle 0^{m}|Q(z)^{\dagger})=\rho \quad \text{ for all } \quad z\in \{0,1\}^{|A|}.
\]

\label{lem:peakedshallowpurification}
\end{lemma}
\begin{proof}
Using Eq.~\eqref{eq:udef} and cyclicity of the trace, we get
\[
\mathrm{Tr}_{A'B}(Q(z)|0^{m}\rangle\langle 0^{m}|Q(z)^{\dagger})=\mathrm{Tr}_{A'B}\big(V_{AB}|0\rangle\langle 0|_{AB}V^{\dagger}_{AB}\otimes |z\rangle\langle z|_{A'}\big)=\rho.
\]

Now write the Schmidt decomposition 
\begin{equation}
V_{AB}|0\rangle_{AB}=\sum_{k}\sqrt{p_k}|\alpha_k\rangle|\beta_k\rangle.
\label{eq:v1}
\end{equation}
so that $\rho=\sum_{k}p_k|\alpha_k\rangle\langle \alpha_k|$. Then  
\[
\langle h_{A}0_{A'B}| Q(h)|0^{m}\rangle= \left(\sum_{k}\sqrt{p_k}  \langle h|_{A}\langle \alpha_k|_{A'}\langle \beta_k|_{B}\right)\left( \sum_{j} \sqrt{p_j} |\alpha_j\rangle_A|\beta_j\rangle_B|h\rangle_{A'}\right)=\langle h|\rho|h\rangle\geq n^{-a},
\]
which completes the proof.
\end{proof}

\section{Peaked shallow circuits \label{sec:ham}}

In this section we prove Theorems \ref{thm:main1} and \ref{thm:mainstate}, following the strategy outlined in the Introduction. As a first step we establish a certain Hamming weight concentration property for probability distributions over $n$-bit strings such that each bit is independent of all but a constant number of others.

\begin{lemma}
\label{lemma:concentration}
Let $P(x)$ be a probability distribution on $n$-bit strings
such that each variable $x_j$ is
 independent of all but at most
$K$ of the others and $\EE_{x\sim P}(x_j)\le 1/2$. 
Let 
\be
\label{peak}
P_{max} =\max_{x\in \{0,1\}^n} P(x)
\ee
and let $\delta\le P_{max}^{18/25}$ be a given error tolerance. 
Then 
\be
\label{eq:concentration}
\mathrm{Pr}_{x\sim P}\left[\sum_{j=1}^n x_j \ge \frac{25}6 (K+1) \log{(1/\delta)}\right]\le \delta.
\ee
Here we use the natural logarithm.
\end{lemma}
Suppose 
$P$ is the output distribution of a constant-depth quantum circuit.
Then the independence condition of the lemma is  satisfied with $K=O(1)$.
Furthermore,  expected values $\EE_{x\sim P} (x_j)$ can be efficiently computed
by restricting the circuit to the lightcone of a single qubit $j$.
By relabeling bit values $0$ and $1$, if necessary, one can ensure that 
$\EE_{x\sim P} (x_j)\le 1/2$ for all $j$.
Suppose $P_{max}\ge n^{-O(1)}$.
Then condition $\delta\le P_{max}^{18/25}$ is satisfied for any polynomial
$\delta=\mathrm{poly}(1/n)$ of sufficiently high degree.
The lemma then implies that 
restricting $P(x)$ to the Hamming ball of diameter
$O(\log{n})$ centered at the all-zero string results in 1-norm error at most $\delta=\mathrm{poly}(1/n)$.

\begin{proof}
Let  $w(x)=\sum_{j=1}^n x_j$ be the Hamming weight of $x$.
We claim that 
\be
\label{average_weight}
\EE_{x\sim P} [w(x)]\le (K+1)\log{(1/P_{max})}.
\ee
Indeed, let $m_j=\EE_{x\sim P} (x_j)$.
Suppose Eq.~(\ref{average_weight}) is false, i.e.
$\sum_{j=1}^n m_j >(K+1)\log{(1/P_{max})}$. Let us show that this leads to a contradiction.
Define a dependency graph $G=(V,E)$ with $n$ vertices such that each vertex
of $G$ represents a variable $x_j$ and edges of $G$ represent dependencies
among the variables. More precisely, we assume that a variable $x_j$
is independent of a subset of variables $S\subseteq V\setminus \{j\}$ whenever
the graph $G$ has no edges between $j$ and $S$.
By assumption, $G$ has degree at most $K$. 
Using a simple randomized argument one can show
that $G$ has an independent set $S\subseteq V$ such that 
\be
\label{setS}
\sum_{j\in S} m_j \ge \frac1{K+1} \sum_{j=1}^n m_j > \log{(1/P_{max})}.
\ee
Let $h\in \{0,1\}^n$ be any bit string such that $P(h)=P_{max}$.
Since all variables $x_j$ with $j\in S$ are independent, one gets
\[
P_{max}=P(h)\le\mathrm{Pr}_{x\sim P}[x_j=h_j \; \mbox{for all} \; j\in S]=
\prod_{j\in S\, : \, h_j=1} m_j \prod_{j\in S\, : \, h_j=0}(1-m_j).
\]
Since we assumed that $m_j\le 1/2$, one has $m_j\le 1-m_j$. Thus
\[
P_{max} \le \prod_{j\in S} (1-m_j) \le \exp{\left[-\sum_{j\in S} m_j \right]}<P_{max}.
\]
Here the last inequality follows from Eq.~(\ref{setS}).
We arrived at a contradiction which proves Eq.~(\ref{average_weight}).
Let 
\[
R:=(K+1)\log{(1/P_{max})}.
\]
Note that  $\EE_{x\sim P} [w(x)]\le R$ due to Eq.~(\ref{average_weight}).
A large deviation bound for a sum of random variables whose dependency graph has a bounded
degree was obtained in~\cite{janson2004large}. Applying Theorem~2.3 of~\cite{janson2004large} gives
\be
\mathrm{Pr}_{x\sim P}\left[ w(x) \ge  R + t\right] \le
\mathrm{Pr}_{x\sim P}\left[ w(x) \ge  \EE_{x\sim P}[w(x)] + t\right] \le \exp{\left( - \frac{8t^2}{25(K+1) (S+t/3)}\right)},
\ee
for any $t\ge 0$,
where $S=\sum_{j=1}^n \mathrm{Var}(x_j) = \sum_{j=1}^n m_j - m_j^2 \le \sum_{j=1}^n m_j\le R$. Thus
choosing $t=(\lambda-1)R$ with $\lambda\ge 1$ one gets
\be
\label{weight_tail_bound1}
\mathrm{Pr}_{x\sim P}\left[ w(x) \ge \lambda R\right]\le 
 \exp{\left( - \frac{24(\lambda-1)^2 R^2}{25(K+1) (2+\lambda)R}\right)}
 \le  \exp{\left( - \frac{6\lambda R}{25(K+1)}\right)}.
\ee
Here we assumed that $\lambda\ge 3$ and used a bound $(\lambda-1)^2/(2+\lambda)\ge \lambda/4$.
Choose 
\be
\label{our_lambda}
\lambda =  \frac{25 \log{(1/\delta)}}{6\log{(1/P_{max})}}.
\ee
Then the righthand side of Eq.~(\ref{weight_tail_bound1}) equals $\delta$ and we get
\[
\mathrm{Pr}_{x\sim P}\left[ w(x) \ge \frac{25}6 (K+1) \log{(1/\delta)}\right]\le \delta.
\]
Condition $\lambda\ge 3$ is equivalent to $\delta\le P_{max}^{18/25}$, see Eq.~(\ref{our_lambda}).
\end{proof}

\subsection{Proof of Theorems \ref{thm:main1} and \ref{thm:mainstate} \label{sec:peaked}}

Suppose we are given a depth-$d$ quantum circuit $U$ composed of $1$- and $2$-qubit gates and
an error tolerance $\epsilon>0$.
We do not assume that $U$ is geometrically local.
Let $K=2^{2d}$ so that each qubit is  lightcone-separated from all but $K$ of other qubits.
Consider an $n$-qubit state $|\psi\ra=U|0^n\ra$ and let $P(x)=|\la x|\psi\ra|^2$ be the output distribution of $U$. 
As discussed above, we can assume wlog that $\EE_{x\sim P}(x_j)\le 1/2$ for all qubits $j$.
Let 
\[
P_{max}=\max_{x\in \{0,1\}^n} P(x).
\]
We assume that $P_{max}$ is sufficiently large so that 
\be
\label{peakness_threshold}
\epsilon \le 3n^{1/2} P_{max}^{9/50}.
\ee
Below we describe classical algorithms that, if they succeed, (a) output a description of an $n$-qubit state $\phi$ such that $|\langle \phi|\psi\rangle|^2\geq 1-\epsilon^2/4$ and (b) sample an $n$-bit string  $x$ from a distribution $P'(x)$ such that $\|P-P'\|_1\le \epsilon$.
If the algorithms do not succeed they output an error flag. The algorithms are guaranteed to succeed if $P_{max}$ is sufficiently large so that Eq.~(\ref{peakness_threshold}) holds. 
 The algorithms have a quasipolynomial runtime 
and space complexity which scale as
\be
\label{runtime_sampling_peaked}
T(n,\epsilon)=n^{O(\log{n} + \log{(1/\epsilon)})}.
\ee
Define
\be
\label{epsilon_vs_delta}
\delta =\frac{\epsilon^4}{144 n^2}.
\ee
Suppose Eq.~(\ref{peakness_threshold}) holds.
Then
$\delta\le P_{max}^{18/25}$. In this case all conditions of
Lemma~\ref{lemma:concentration} are satisfied and we get
\be
\label{overlap_low_weight}
\sum_{x\, : \, w(x)\le W} P(x) \ge 1-\delta, 
\ee
where $w(x)=\sum_{j=1}^n x_j$ is the Hamming weight and
\be
\label{weight_cutoff}
W= (25/6)(K+1) \log{(1/\delta)}.
\ee
The state $\psi$ can be described as the unique  largest eigenvector with the eigenvalue one of a local parent Hamiltonian 
\be
H = \frac1n\sum_{j=1}^n U |0\ra\la 0|_j U^\dag.
\ee
Here $|0\ra\la 0|_j$ is the projector $|0\ra\la 0|$ acting on the $j$-th qubit tensored with the identity on all remaining qubits.
Let $\lambda_r(H)$ be the $r$-th largest eigenvalue of $H$. 
Since $H$ has the same spectrum as a diagonal Hamiltonian $(1/n)\sum_{j=1}^n |0\ra\la 0|_j$,
we have
\be
\lambda_1(H)=1 \quad \mbox{and} \quad \lambda_2(H)=1-1/n.
\ee
\begin{lemma}
\label{lemma:truncation}
Let $\Pi$ be a projector onto the Hamming ball of radius $W$ centered at $0^n$,
\be
\Pi = \sum_{x\, : \, w(x)\le W} |x\ra\la x|,
\ee
where $W$ is the cutoff value defined in Eq.~(\ref{weight_cutoff}).
Consider a projected parent Hamiltonian  $G = \Pi H \Pi$ and let $\phi$ be an eigenvector of $G$ with the largest
eigenvalue $\lambda_1(G)$. If this eigenvalue satisfies
\begin{equation}
\lambda_1(G)\geq 1-3\sqrt{\delta},
\label{eq:lambda1G}
\end{equation}
then we have
\be
|\la \phi |\psi\ra|^2 \ge  1-3n\sqrt{\delta}.
\label{eq:stateapprox}
\ee
Moreover, Eq.~\eqref{eq:lambda1G} holds whenever Eq.~(\ref{peakness_threshold})
does.
\label{lem:projectedG}
\end{lemma}
\begin{proof}
Write
$\psi = \psi_0+\psi_1$ for some vectors $\psi_0,\psi_1$ such that 
$\Pi \psi_1=\psi_1$ and $\Pi\psi_0=0$.
Then
\begin{align*}
\lambda_1(G)\ge & \la \psi|G|\psi\ra =  \la \psi_1|H|\psi_1\ra
=\la \psi|H|\psi\ra - \la \psi_0|H|\psi_0\ra - 2\mathrm{Re}(\la \psi_1|H|\psi_0\ra)\\
\ge & 1- \|\psi_0\|^2 - 2\|\psi_0\|
\end{align*}
If Eq.~(\ref{peakness_threshold}) holds then from Eq.~(\ref{overlap_low_weight}) one gets $\|\psi_1\|^2 =\la \psi|\Pi|\psi\ra \ge 1-\delta$ and
$\|\psi_0\|^2  = 1-\|\psi_1\|^2 \le \delta$. Therefore Eq.~(\ref{peakness_threshold}) implies
\be
\label{largest_eig_upper}
\lambda_1(G)\ge 1-\delta - 2\sqrt{\delta} \ge 1-3\sqrt{\delta},
\ee
assuming that $0\le \delta \le 1$.  

To complete the proof we now show that Eq.~\eqref{largest_eig_upper} implies Eq.~\eqref{eq:stateapprox}. Write
\be
H = |\psi\ra\la \psi | + H^\perp
\ee
where
\be
\label{Hperp}
H^\perp \psi=0 \quad \mbox{and} \quad 
0\le H^\perp \le (1-1/n) I.
\ee
Let $\phi$ be an eigenvector of $G$
with the eigenvalue $\lambda_1(G)$.
It can be written as
$\phi = \alpha \psi + \psi^\perp$
where $\psi^\perp$ is some vector orthogonal to $\psi$. We can assume wlog that $\alpha\ge 0$.
Normalization of $\phi$ gives $\|\psi^\perp\|^2 = 1-\alpha^2$.
Clearly, $\Pi \phi = \phi$. Thus
\be
\label{largest_eig_lower}
\lambda_1(G)=\la \phi|G|\phi\ra = \la \phi|H|\phi\ra = \alpha^2 + \la \phi|H^\perp|\phi\ra
=  \alpha^2+\la \psi^\perp |H^\perp |\psi^\perp\ra
\le \alpha^2 + (1-\alpha^2) (1-1/n).
\ee
Here the fourth equality and the last inequality use Eq.~(\ref{Hperp}).
From Eqs.~(\ref{largest_eig_upper},\ref{largest_eig_lower}) one gets
\be
1- 3\sqrt{\delta} \le \lambda_1(G)\le  \alpha^2 + (1-\alpha^2) (1-1/n)
\ee
which implies
\be 
|\la \phi|\psi\ra|^2 = \alpha^2 \ge 1-3n\sqrt{\delta}.
\ee
\end{proof}

Now let us describe our algorithms. The first step is to compute the largest eigenvalue $\lambda_1(G)$ and an eigenvector $\phi$ of $G$ with this eigenvalue. To this end, consider a Hilbert space $\calH$ whose basis vectors are $n$-bit strings with the Hamming weight at most $W$.
Here $W$ is the cutoff value defined in Eq.~(\ref{weight_cutoff}). We have
\be
D=\mathrm{dim}(\calH) = \sum_{j=0}^W {n\choose j} \le n^W = n^{O(\log{(1/\delta)})} = n^{O(\log{n} + \log{(1/\epsilon)})}.
\label{eq:D}
\ee
To compute $\lambda_1(G)$ and an eigenvector of $G=\Pi H \Pi$ with the largest eigenvalue we can restrict $G$ onto $\calH$.
Any matrix element of $G$ in the chosen basis of $\calH$ coincides with some matrix element of $H$
in the standard basis. Since $H$ is a local Hamiltonian with $n$ local terms, any matrix element of $H$ can be computed
in time $O(n)$. Now one can use standard linear algebra algorithms
 to compute an eigenvector of $G$ with the
largest eigenvalue in time $O(nD^3)$.

The algorithm for approximating the output state proceeds as follows. If $\lambda_1(G)$ satisfies Eq.~\eqref{largest_eig_upper} then we output the state $\phi$; otherwise we output an error flag. By Lemma \ref{lem:projectedG}, if we output the state $\phi$,  it satisfies 
\be
\label{psi_phi_overlap}
|\langle\phi|\psi\rangle|^2\geq 1-3n\sqrt{\delta}\geq 1-\epsilon^2/4.
\ee
To approximately sample from the output distribution,  we first construct the above approximation $\phi$ to the output state (or output an error flag). If this step succeeds, then we continue by computing all probabilities $P'(x)=|\la x|\phi\ra|^2$ and sample $P'(x)$ in time $O(D)$. 

The standard inequality between trace distance and fidelity gives
\be
\| P -P'\|_1
=\sum_x |P(x) - P'(x)|
\le 2\sqrt{1-|\la \phi|\psi\ra|^2} 
\le \epsilon.
\ee
Here the last equality uses  Eq.~(\ref{psi_phi_overlap}).
The above argument proves a slightly stronger version of Theorem~\ref{thm:mainstate} 
which we state below.

\begin{theorem}
\label{thm:mainstate_generalized}
There exists a classical algorithm which takes as input an $n$-qubit constant-depth circuit $U$ and a precision parameter $\epsilon>0$. If the algorithm succeeds then it outputs a classical description of an $n$-qubit state $|\phi\rangle$ specified by its 
$n^{O(\log{n} + \log{(1/\epsilon)})}$
nonzero amplitudes in the computational basis such that $|\langle \phi|U|0^n\rangle|^2\geq 1-\epsilon^2/4$; otherwise it outputs an error flag. The algorithm is guaranteed to succeed if $U$ is sufficiently peaked such that\ $\epsilon \le 3n^{1/2} P_{max}^{9/50}$, where
\[
P_{max}=\max_{x\in \{0,1\}^n} |\la x|U|0^n\ra|^2.
\]
 The runtime of the algorithm is $T(n,\epsilon)=n^{O(\log{n} + \log{(1/\epsilon)})}$.
\label{thm:mainstate_generalized}
\end{theorem}
Theorem~\ref{thm:mainstate} is obtained as a special case of
Theorem~\ref{thm:mainstate_generalized} with $P_{max}=1/n^a$.
We shall use Theorem~\ref{thm:mainstate_generalized}  later on in Section~\ref{sec:kdim}.

\subsection{Software implementation}
\label{sec:numerics}

\begin{figure}[t]
\subfloat{\includegraphics[scale=0.45]{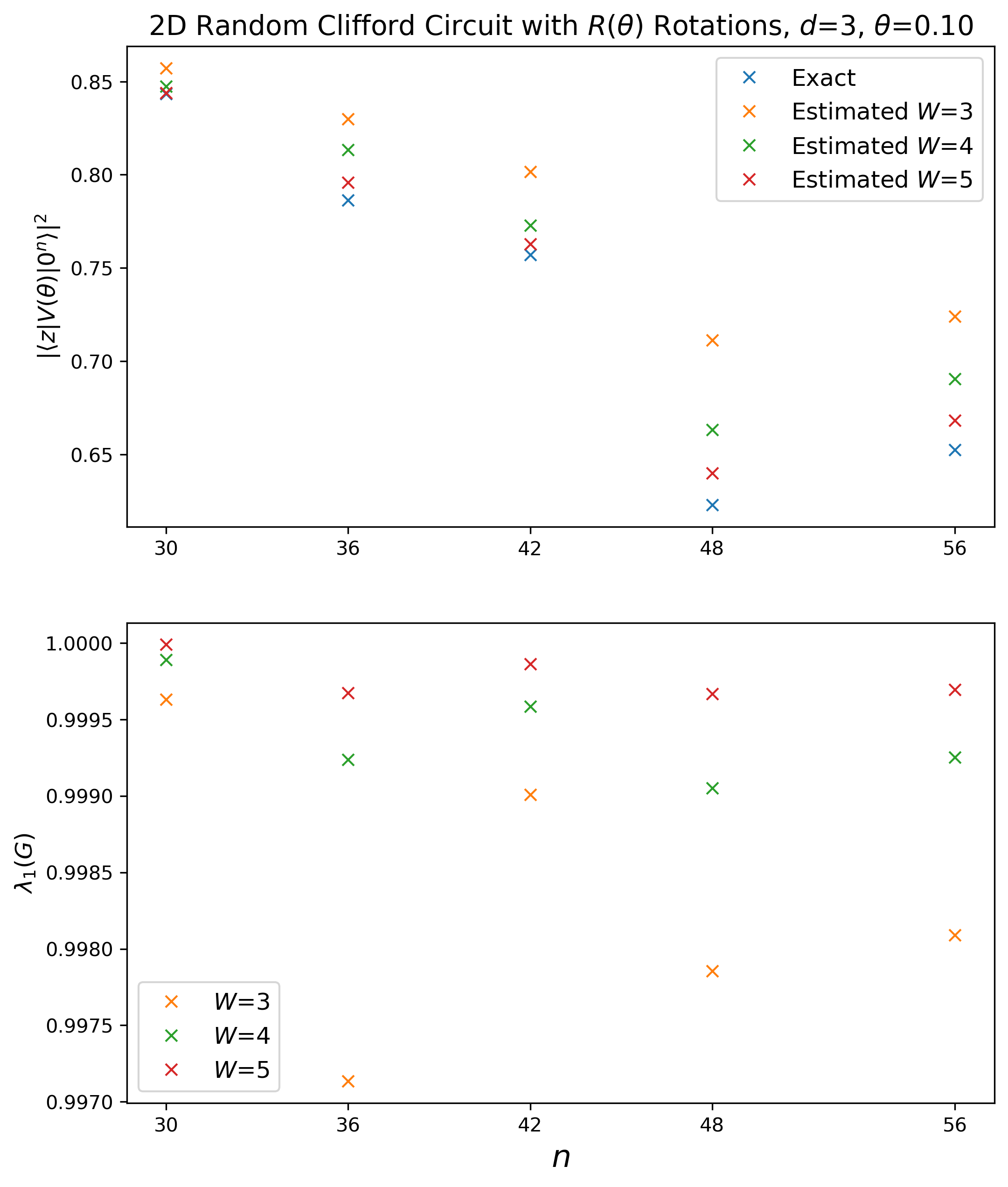}}
\qquad
\subfloat{\includegraphics[scale=0.45]{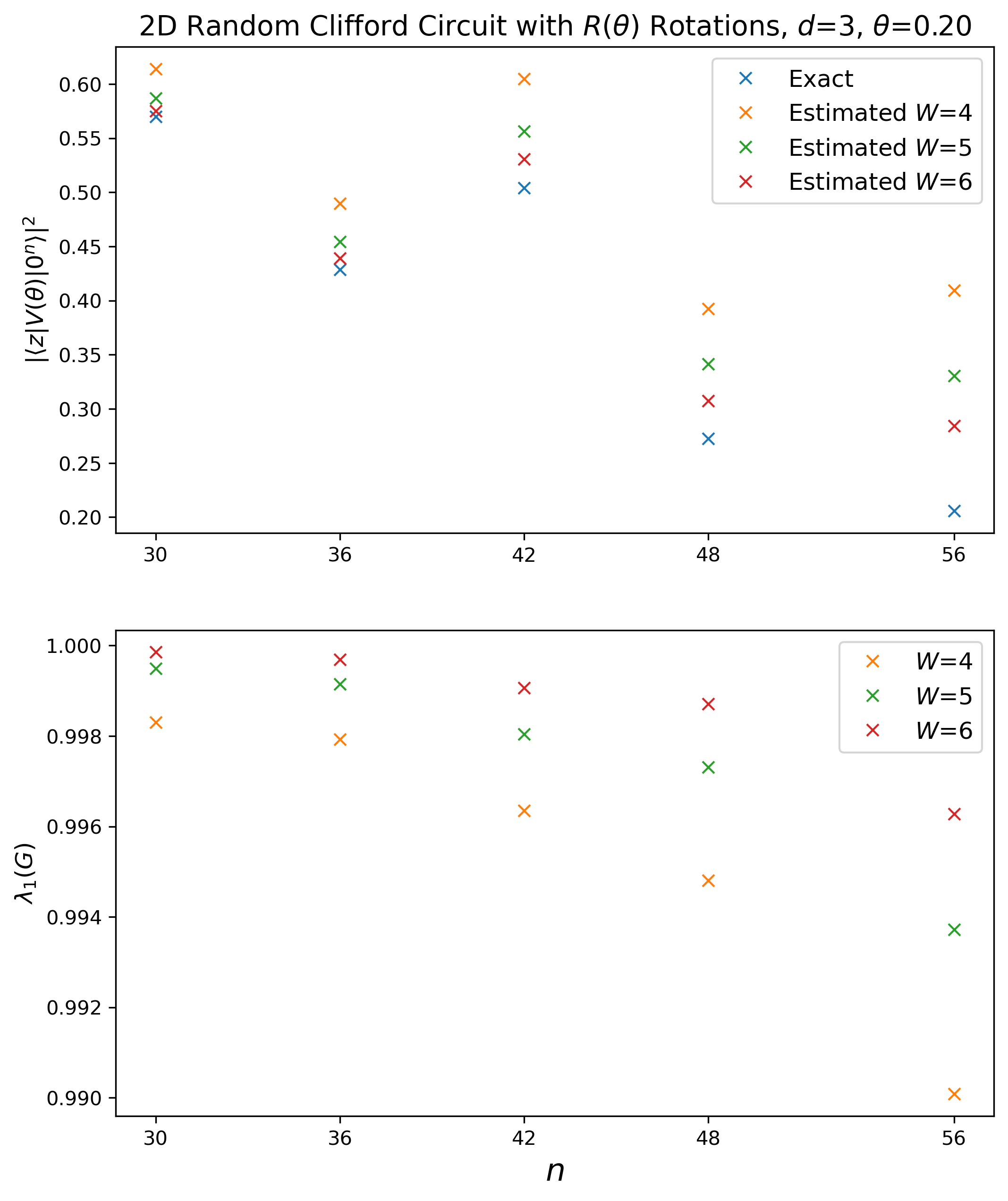}}
\caption{Numerical simulation of 2D random Clifford circuits interspersed with $R(\theta)$ rotations; the number of qubits $n$ ranges from $30$ to $56$, and all simulated $U(\theta)$ circuits have depth $d=3$ (so the $V(\theta)$'s all have depth $5$); $\lambda_1(G)$ is the largest eigenvalue of the $G$ matrix defined in Lemma \ref{lemma:truncation}, and $W$ is the radius of Hamming ball centered at $z$; the algorithm is self-certifying since the simulation error is inversely proportional to $\lambda_1(G)$.}
\label{fig:numerics}
\end{figure}

Here we investigate the practical performance of the peaked shallow circuit simulation algorithm from the previous section. In particular, we implement our peaked constant-depth algorithm to simulate $n$-qubit 2D random Clifford circuits interspersed with $Z$-rotations. We consider quantum circuits $U(\theta)$ composed of $2$-qubit gates arranged in a 2D open boundary condition brickwork circuit architecture. Every gate of $U(\theta)$ is a uniformly random $2$-qubit Clifford gate followed by single qubit gates $R(\theta)\otimes R(\theta)$ where
\begin{equation}
R(\theta)=\begin{bmatrix}e^{-i\frac{\theta}{2}}&0\\0&e^{i\frac{\theta}{2}}\end{bmatrix}.
\end{equation}
Define $V(\theta)=U^\dagger(\theta)Z^{\otimes n}U(\theta)$. Since $U(0)$ is a Clifford circuit, $V(0)$ is an $n$-qubit Pauli operator, so there exists $z\in\{0,1\}^n$ such that $|\langle z|V(0)|0^n\rangle|=1$. Hence, for small values of $\theta$, $V(\theta)$ will be a peaked circuit (whose peak is $z$), and $V(\theta)$ becomes less peaked as $\theta$ increases. We implement our algorithm to simulate instances of $V(\theta)$ for $\theta\in\{0.1,0.2\}$, $n$ ranging from $30$ to $56$, and using increasing values of Hamming ball radius $W$. 
The algorithm outputs an estimate $p=P'(z)$ of the probability $|\langle z|V(\theta)|0^n\rangle|^2$.
Recall that $P'(z)=|\la z|\phi\ra|^2$, where $\phi$
is the state approximating $V(\theta)|0^n\ra$ constructed in Lemma~\ref{lemma:truncation}. Let $G$ be the projected Hamiltonian defined in Lemma~\ref{lemma:truncation}. 
Numerical results from Fig. \ref{fig:numerics} confirm that the algorithm's simulation error is controlled by $\lambda_1(G)$: the closer $\lambda_1(G)$ is to $1$, the closer the approximate values of $|\langle z|V(\theta)|0^n\rangle|^2$ outputted by the algorithm are to the exact values. Indeed, the proof of Lemma \ref{lemma:truncation} implies the upper bound
\begin{equation}
|p - |\langle z|V(\theta)|0^n\rangle|^2|
=|P'(z)-P(z)|
\le
\lVert P'-P\rVert_1\leq 2\sqrt{n(1-\lambda_1(G))}.
\end{equation}

Although the algorithm enjoys favourable scaling in $n$, its scaling in $W$ makes it prohibitive to try $W\geq 7$ for simulating circuits with $n=56$. The observed simulation error is also very sensitive to how peaked the circuits are; comparing to the results for $\theta=0.1$, the error margins widen for $\theta=0.2$ where the circuits being simulated are not as peaked. Nonetheless, the value of $\lambda_1(G)$ can still be used to gauge simulation accuracy. In addition, it is worth noting that even for the largest $n=56$ circuits considered here, the amplitude $\langle z|V(\theta)|0^n\rangle$ (including the sign) can be computed exactly using off-the-shelf tensor network simulators in seconds \cite{gray2018quimb}. However, our algorithm finds at once a whole distribution close to the exact one (i.e. an eigenvector corresponding to $\lambda_1(G)$), and it could be much slower to use tensor network based methods to construct a distribution-encoding vector one entry at a time. Therefore, it may be advantageous to adopt our algorithm for applications which require full distributional information (for example, approximate sampling or computation of entanglement measures).

\section{Geometrically local peaked shallow circuits \label{sec:kdim}} 
Let $U$ be a depth $d=O(1)$ circuit in a $D$-dimensional geometry with $n$ qubits arranged in a $n^{1/D}\times n^{1/D}\times\ldots n^{1/D}$ grid. Define the output state $|\psi\rangle=U|0^n\rangle$ and output distribution
\[
P(x)=|\langle x|\psi\rangle|^2.
\]
As in the previous section, we shall assume without loss of generality that 
\begin{equation}
\mathbb{E}_{x\sim P}[x_j] \leq 1/2,
\label{eq:mj}
\end{equation}
where $x_j$ is the $j$th bit of $x$. We shall be interested in the case where the circuit is $a$-peaked for a given $a\in \mathbb{Z}_{\geq 0}$, i.e., 
\begin{equation}
P(h)\geq n^{-a}
\label{eq:suppose}
\end{equation}
for some $h\in \{0,1\}^n$. We do not assume that $h$ is given. However, Eqs.~(\ref{eq:mj},\ref{eq:suppose}) along with Lemma \ref{lemma:concentration} (setting $\delta=n^{-2a}$) imply that the Hamming weight $w(h)\equiv \sum_{j=1}^{n}h_j$ satisfies
\begin{equation}
w(h)\leq \frac{50a}{6}(K+1)\log(n).
\label{eq:weightofh}
\end{equation}
where $K=O(1)$ is a known constant. Below we give a classical sampling algorithm with the following runtime.

\begin{theorem*}[Thm.~\ref{thm:main2}, restated]
Suppose we are given a $D$-dimensional geometrically local, shallow quantum circuit $U$ as described above,  a positive integer $a\in \mathbb{Z}_{\geq 0}$, and $\epsilon= \Omega(1/\mathrm{poly}(n))$. There exists a classical algorithm which either returns a sample from a distribution $P'$ over $n$-bit strings such that $\|P'-P\|_1\leq \epsilon$, or outputs an error flag. The algorithm always returns a sample if Eq.~\eqref{eq:suppose} holds. The runtime of the algorithm is upper bounded as
\begin{equation}
\begin{cases}
O(\mathrm{poly}(n)) & \text{ if } D=2\\
n^{O(\log \log{n})} & \text{ if } D>2.\\
\end{cases}
\end{equation} 
\end{theorem*} 
In particular, suppose we are given some $\epsilon=\Omega(1/\mathrm{poly}(n))$. Let $\mathcal{M}(\cdot)$ be the CPTP map corresponding to measuring all qubits in the computational basis. 

\begin{theorem}
If Eq.~\eqref{eq:suppose} holds then there exists a Hermitian matrix  $K$ satisfying
\begin{equation}
\|\mathcal{M}(K)-\mathcal{M}(|\psi\rangle\langle \psi|)\|_1=\sum_{y\in \{0,1\}^n} |\langle y|K|y\rangle-|\langle y|\psi\rangle^2| \leq \epsilon,
\label{eq:rhoerror}
\end{equation}
and such that there exists a classical algorithm to exactly compute ``marginals" $\mathrm{Tr}(K |x\rangle \langle x|_S\otimes I)$ for any 
 $S\subseteq [n]$ and $x\in \{0,1\}^{|S|}$. The runtime of the algorithm is $O(\mathrm{poly}(n))$ if $D=2$ and $n^{O(\log \log{n})}$ if $D>2$. If the algorithm is run with input unitary $U$ that does not satisfy Eq.~\eqref{eq:suppose} it will either succeed, in which case Eq.~\eqref{eq:rhoerror} holds, or it will output an error flag.
\label{thm:main}
\end{theorem}
Theorem \ref{thm:main2} follows  from Theorem \ref{thm:main} along with the fact that the standard linear-time reduction from sampling to computing marginals (``bit-by-bit") is robust to error in the 1-norm, even when the approximating vector can take negative values, as shown in Ref.~\cite{sparse}.

\begin{lemma}[Lemma 10 of Ref.~\cite{sparse}]
Suppose $p$ is a probability distribution over $n$-bit strings, $\delta=\Omega(1/\mathrm{poly}(n))$, and $p':\{0,1\}^n\rightarrow \mathbb{R}$ satisfies $\|p'-p\|_1\leq \delta<1$. Suppose further that  there is a classical algorithm with runtime $T(n)$ that, given $0\leq j\leq n$ and $x\in \{0,1\}^{j}$ exactly computes 
\[
\sum_{y\in \{0,1\}^{n-j}}p'(xy).
\]
Then there is a classical algorithm with runtime $nT(n)$ which samples from a probability distribution $q$ such that $\|q-p\|_1\leq 4\delta/(1-\delta)$.
\end{lemma}

In the proof of Theorem \ref{thm:main} for $D=2$ we will construct $K$ explicitly as a pseudomixture (real linear combination of density matrices) of tensor products of matrix product states. For $D>2$, $K$ is constructed as a pseudomixture of tensor products of states on subgrids each consisting of $O(\log^{D}{n})$ qubits. The remainder of this section is concerned with establishing Theorem \ref{thm:main}. 
\subsection{Dissection along a side}
Our proof technique uses some of the ideas from Ref.~\cite{coudron}. In particular, the method is based on choosing one of the sides of the $D$-dimensional grid and identifying \textit{heavy slices} which are subgrids of qubits of constant thickness along the chosen side, that have a high chance of being measured in the all-zeros state. However our method differs substantially from Ref.~\cite{coudron} in the way that heavy slices are used in the proof. Whereas Ref.~\cite{coudron} uses a heavy slice to divide the grid into left and right subproblems (reducing to problems of half the size), here we show how to use a 1D regular arrangement of heavy slices to reduce our $D$-dimensional problem to one that is quasi-$(D-1)$-dimensional. More precisely, the type of quantum states we will encounter will be reduced density matrices of $|\psi\rangle=U|0^n\rangle$ of the following type.

\begin{dfn}
Let us say that $\rho$ is an $(n,k,D)$-state, if it is equal to the reduced density matrix of $\psi$ on a $D$-dimensional rectangular grid of qubits with $k$ sides of length $n^{\frac{1}{D}}$ and $D-k$ sides of length $O(\log{n})$.
\end{dfn}

Suppose $\rho$ is an $(n,k,D)$-state with $1\leq k\leq D$.  Label the qubits in the support of $\rho$ as $(i_1,i_2,\ldots, i_D)$ where $1\leq i_j\leq S_j$ and $S_j$ is the length of the $j$th side of the grid of qubits. Let us suppose without loss of generality that the first side $S_1$ has length $n^{1/D}$. It will be convenient to refer to subsets of the qubits of the form
\[
 \{(i_1,i_2,\ldots, i_D): \alpha\leq i_1\leq \beta, \text{ and } 1\leq i_j\leq S_j \text{ for } 2\leq j\leq D\}.
\]
which we call a \textit{rectangular region} of width $\beta-\alpha$. 
Given a qubit register $A$, we shall write $|0\ra_A$ for the all-zeros basis state of $A$.

\begin{dfn}
Suppose $\rho$ is an $(n,k,D)$-state.  A heavy slice $s$ is any rectangular region of width $w=2d+1$ satisfying
\[
\mathrm{Tr}(\rho (|0\rangle\langle 0|_s\otimes I))\rangle\geq 0.99.
\]
\end{dfn}
Note that the definition of a heavy slice depends on the support of the reduced density matrix $\rho$ as well as the circuit $U$, though this dependence is suppressed in our terminology.

A heavy slice $s$ partitions the support of $\rho$ into three rectangular regions consisting of $s$ and the regions $L$ and $R$ to the left and right of $s$, respectively. The condition that the width of the slice satisfies $w= 2d+1$ ensures that the backwards lightcones of $L$ and $R$ are disjoint. This means in particular that when we trace out any heavy slice $s$ we obtain a tensor product
\[
\mathrm{Tr}_s (\rho)= \sigma_L \otimes \sigma_R.
\]
for some states $\sigma_L, \sigma_R$ on the left and right rectangular regions, respectively.

If $s$ is a heavy slice then there is a good (at least 99\%) chance of measuring $0$ for all qubits in $s$. Therefore measuring the qubits in $s$ in the computational basis is approximated by simply tracing out the qubits in $s$ and replacing them with all zeros, as we now show. For any subset of qubits $\Omega$ define CPTP maps
\[
M_\Omega (\kappa)\equiv \sum_{z\in \{0,1\}^{|\Omega|}} (|z\rangle\langle z|_\Omega\otimes I)\kappa (|z\rangle\langle z|_\Omega\otimes I) \qquad \text{ and } \qquad E_\Omega(\kappa)\equiv \mathrm{Tr}_\Omega(\kappa)\otimes |0\rangle\langle 0|_\Omega.
\]
\begin{lemma}
Suppose $\rho$ is an $(n,k,D)$-state.  Let $s$ be a heavy slice. Then
\[
\|M_s(\rho)-E_s(\rho) \|_1 \leq 0.02
\]
\label{lem:heavy}
\end{lemma}
\begin{proof}
Using the definitions we can write
\[
M_s(\rho)-E_s(\rho)  =\sum_{z\in \{0,1\}^{|s|}, z\neq 0^{|s|}} (\langle z|_s\otimes I)\rho (|z\rangle_s\otimes I) \otimes \left(|z\rangle\langle z|_s-|0\rangle\langle 0|_s\right)
\]
Since $\|\left(|z\rangle\langle z|_s-|0\rangle\langle 0|_s\right)\|_1=2$ whenever $z\neq 0^{|s|}$, we have
\begin{align*}
\|M_s(\rho)-E_s(\rho) \|_1& \leq \sum_{z\in \{0,1\}^{|s|}, z\neq 0^{|s|}} 2\| (\langle z|_s\otimes I)\rho (|z\rangle_s\otimes I)\|_1\\
&=2\sum_{z\in \{0,1\}^{|s|}, z\neq 0^{|s|}} \mathrm{Tr}(\rho (|z\rangle\langle z|_s\otimes I))\\
&=2\mathrm{Tr}(\rho (I-|0\rangle\langle 0|_s)) \\
&\leq 0.02,
\end{align*}
where in the last inequality we used the fact that $s$ is a heavy slice.
\end{proof}

We now show that Eqs.~(\ref{eq:suppose}, \ref{eq:mj}) imply the existence of such heavy slices. In the following lemma, we show how to find a large number of them, arranged in a certain regular pattern. To describe this pattern we will refer to a partition of the grid into $2T$ equal-width rectangular regions as shown in Fig.~\ref{fig:partition} for the special case $D=2$. Below we will choose the width of the rectangular regions to satisfy
\[
W=\Theta (d\log{n})
\]
so that the total number of regions is
\[
T=\Theta\left(\frac{n^{1/D}}{d\log({n})}\right).
\]

\begin{figure}
\centering
\begin{tikzpicture}[scale=0.5]
\draw (0,0) rectangle (10,10);
\draw (12,0) rectangle (16,10);

\draw node at (11,5){$\ldots$};
\draw node at (1,5) {$V^1$};
\draw[fill=lightgray] (2,0) rectangle (4,10);
\draw node at (3,5) {$H^1$};
\draw node at (5,5) {$V^2$};
\draw[fill=lightgray] (6,0) rectangle (8,10);
\draw node at (7,5) {$H^2$};
\draw node at (9,5) {$V^3$};
\draw[fill=lightgray] (14,0) rectangle (16,10);
\draw node at (13,5) {$V^T$};
\draw node at (15,5) {$H^T$};

\draw[thick,<->] (0,-1)--(2,-1);
\draw[thick,<->] (2,-1)--(4,-1);
\draw node at (1,-2){$W=\frac{\sqrt{n}}{2T}$};

\draw[thick,<->] (0,11)--(16,11);
\draw node at (8,12){$\sqrt{n}$};
\draw[thick, <->] (-1,0)--(-1,10);
\draw node at (-2,5){$\sqrt{n}$};
\end{tikzpicture}
\caption{A partition of the grid of qubits into equal-width rectangular regions, depicted here in the two-dimensional case. Lemma \ref{lem:lotsofslices} describes how we can find lightcone-separated heavy slices within each region $H^{j}$ for $1\leq j\leq T$.\label{fig:partition}}
\end{figure}
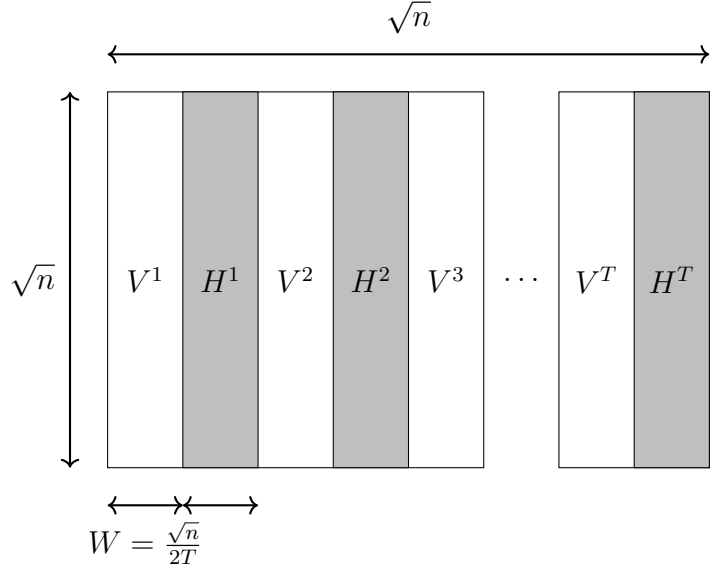

\begin{lemma}
Suppose Eq.~\eqref{eq:suppose} holds, and that $\rho$ is an $(n,k,D)$-state. Consider partitioning the qubits into equal-sized rectangular regions as shown in Fig.~\ref{fig:partition}. There exist $\ell=\Theta(\log{n})$ heavy slices $s_1^{j}, s_2^{j},\ldots, s_\ell^{j}$ within each region $H^{j}$ for $1\leq j\leq T$, where
\begin{equation}
\ell\geq \log(4n^{1/D}\epsilon^{-1}).
\label{eq:lcond}
\end{equation}
The heavy slices are all pairwise lightcone-separated. If $D=k=2$ there exists an efficient classical algorithm which outputs the heavy slices. 
\label{lem:lotsofslices}
\end{lemma}
\begin{proof}
Let $W=\frac{n^{1/D}}{2T}$ be the width of each region $H^{j}$. Recall that the heavy slices we aim to find must have width $2d+1$ and must be lightcone-separated from other heavy slices. Let us fix $N=\Theta(W/d)$ candidate heavy slices $C^{j}_1,C^{j}_2,\ldots, C^{j}_N\subseteq H^{j}$ all of which satisfy these two properties (required width and separation from the other candidates). Recall that $h$ is the ``heavy string" satisfying Eq.~\eqref{eq:suppose}. For each $j\in [T]$ and $i\in [N]$, let $h_{ij}$ be the substring of $h$ consisting of the bits in region $C^{j}_i$.  Since all the candidates are lightcone-separated we have, for each $j\in [T]$, 
\[
\prod_{i=1}^{N}\mathrm{Tr}(\rho |h_{ij}\rangle\langle h_{ij}|_{C^{j}_i})=\langle \psi| \left(\prod_{i=1}^{N}|h_{ij}\rangle\langle h_{ij}|_{C^{j}_i}\right)|\psi\rangle \geq P(h)\geq \frac{1}{n^{a}},
\]
where we used Eq.~\eqref{eq:suppose}. To ease notation let us fix $j\in [T]$ and write $p_i=\mathrm{Tr}(\rho |h_{ij}\rangle\langle h_{ij}|_{C^{j}_i})$. Then taking the log of the above gives
\[
\frac{1}{N}\sum_{i=1}^{N}\log\left(\frac{1}{p_i}\right) \leq \frac{a\log{n}}{N}.
\]
By Markov's inequality, there is a subset $S\subseteq [N]$ of size $|S|\geq N/2$ such that 
\[
\log\left(\frac{1}{p_i}\right)\leq \frac{2a\log{n}}{N} \qquad \text{ for all } i\in S.
\]

Our heavy slices within region $H^{j}$ will be the candidates $C^{j}_i$ with $i\in S$. We choose $W=\Theta(d\log{n})$ large enough so that $N=\Theta(W/d)$ satisfies $\frac{2a\log{n}}{N}\leq \log(1/0.99)$ and so that $N/2\geq \log(4n^{1/D}\epsilon^{-1})$.  The first condition ensures that $p_i\geq 0.99$. From Eq.~\eqref{eq:mj} we see that this implies that $h_{ij}$ is the all-zeros string whenever $i\in S$, i.e., 
\[
\mathrm{Tr}(\rho |0\rangle\langle 0|_{C^{j}_i}) \geq 0.99 \quad \text{for all} \quad i\in S,
\]
which shows that our candidates $C^{j}_i$ with $i\in S$ are heavy slices. The condition $ N/2\geq \log(4n^{1/D}\epsilon^{-1})$ ensures that the number of candidates $\ell=|S|\geq N/2$ within each region $H^{j}$ satisfies Eq.~\eqref{eq:lcond}.

So far we have shown the existence of enough heavy slices amongst our candidates.  If $D=k=2$ then we can efficiently test whether or not a given candidate is a heavy slice, since the circuit composed of all the gates in the backward lightcone of a given candidate is a depth-$d$ circuit on a rectangular strip of size $O(d)\times \sqrt{n}$, and can be simulated efficiently using matrix product state techniques.
\end{proof}

From Lemma \ref{lem:lotsofslices} we infer the existence of lots of lightcone-separated heavy slices. We will use them along with the following extension of Lemma \ref{lem:heavy} (whose proof is quite similar).
\begin{lemma}
Suppose $\rho$ is an $(n,k,D)$-state. Let $s_1, s_2,\ldots, s_m$ be any set of lightcone-separated heavy slices. For any $C\subseteq [m]$ we have 
\[
\bigg\|\prod_{k\in [m]\setminus C} M_{s_k}\prod_{i\in C} \left(M_{s_i}-E_{s_{i}}\right)(\rho) \bigg\|_1 \leq 0.02^{|C|}
\]
(the above product should be interpreted as the composition of commuting superoperators).
\label{lem:heavy2}
\end{lemma}
\begin{proof}
We have
\begin{align*}
& \prod_{k\in [m]\setminus C} M_{s_k}\prod_{i\in C} \left(M_{s_i}-E_{s_{i}}\right)(\rho)\\&=\sum_{z_1,z_2,\ldots z_m} \left(\otimes_{i}\langle z_i|_{s_i} \right)\rho \left(\otimes_{i}|z_i\rangle_{s_i}\right)\bigotimes_{i\in C}\left(|z_i\rangle\langle z_i|_{s_i}-|0\rangle\langle 0|_{s_i}\right)\bigotimes_{k\in [m]\setminus C} |z_k\rangle\langle z_k|_{s_k}
\end{align*}
and therefore
\begin{align}
\bigg\|\prod_{k\in [m]\setminus C} M_{s_k}\prod_{i\in C} \left(M_{s_i}-E_{s_{i}}\right)(\rho) \bigg\|_1&\leq \sum_{z_1,z_2,\ldots, z_m: \; z_i\neq 0^{|s_i|} \text{ for } i\in C} \|\left(\otimes_{i}\langle z_i|_{s_i} \right)\rho \left(\otimes_{i}|z_i\rangle_{s_i}\right)\|_1 2^{|C|}\\
&=2^{|C|} \mathrm{Tr}(\rho\prod_{i\in C}(I-|0\rangle\langle 0|_{s_i}))\\
&\leq 0.02^{|C|}.
\end{align}
\end{proof}

Now for each $j\in [T]$ let us define a superoperator that describes measurement (in the computational basis) of all qubits in all the heavy slices in region $H^j$ from Lemma \ref{lem:lotsofslices}
\[
M^{j}= \prod_{1\leq i\leq \ell} M_{s^{j}_{i}}.
\]
Write
\[
M^{j}=A^{j}+B^{j},
\]
where
\begin{equation}
A^{j}\equiv \sum_{\omega\subseteq [\ell], \omega \neq \emptyset} (-1)^{|\omega|+1} \prod_{i\in [\ell]\setminus\omega} M_{s^{j}_{i}} \prod_{i\in \omega} E_{s^{j}_{i}}
\label{eq:Adef}
\end{equation}
and 
\[
B^{j}=M^{j}-A^{j}=\prod_{i=1}^{\ell} \left(M_{s^j_i}-E_{s^j_{i}}\right).
\]
 Note that a direct application of Lemma \ref{lem:heavy2} gives, for any $\Omega \subseteq [T]$ and $S\subseteq \Omega$,
\begin{equation}
\bigg\|\prod_{j\in \Omega \setminus S} M^{j} \prod_{k\in S} B^{k} (\rho)\bigg\|_1\leq 0.02^{\ell |S|}.
\label{eq:mb}
\end{equation}

\begin{lemma}
Suppose $\rho$ is an $(n,k,D)$-state, and that $s_1^{j},s_2^{j},\ldots, s_\ell^{j}$  for $1\leq j\leq T$ are heavy slices with the properties described in Lemma \ref{lem:lotsofslices}. Then, for any $\Omega\subseteq [T]$ we have
\[
\left\|\big(\prod_{j\in \Omega} A^j-\prod_{j\in \Omega}M^j\big)(\rho)\right\|_1 \leq  \epsilon/2.
\]
\label{lem:1norm}
\end{lemma}
\begin{proof}
\[
\prod_{j\in \Omega} A^j(\rho)=\prod_{j\in \Omega} (M^{j}-B^{j}) (\rho)
=\prod_{j\in \Omega}M^j(\rho) +\sum_{S\subseteq \Omega, S\neq \emptyset} (-1)^{|S|}\prod_{j\in \Omega\setminus S} M^{j} \prod_{k\in S} B^{k} (\rho).
\]
Now using the above, the triangle inequality, and Eq.~\eqref{eq:mb} we arrive at
\begin{align}
\left\|\big(\prod_{j\in \Omega} A^j-\prod_{j\in \Omega}M^j\big)(\rho)\right\|_1&\leq \sum_{S\subseteq \Omega, S\neq \emptyset} 0.02^{\ell |S|} \nonumber\\
&=(1+0.02^{\ell})^{|\Omega|}-1\nonumber\\
&\leq (1+0.02^{\ell})^{T}-1\\
&=0.02^{\ell}\sum_{j=0}^{T-1} (1+0.02^{\ell})^{j}\nonumber\\
&\leq 0.02^{\ell} T(1+0.02^{\ell})^{T}\nonumber\\
&\leq T\cdot 0.02^{\ell} e^{T \cdot 0.02^{\ell}}. \label{eq:epsbound}
\end{align}
Now since $\ell$ satisfies Eq.~\eqref{eq:lcond}, we have $0.02^{\ell}\leq e^{-\ell}\leq \epsilon/(4n^{1/D})$. We also have $T\leq n^{1/D}$ (see Fig.~\ref{fig:partition}). Plugging these into Eq.~\eqref{eq:epsbound} we see that
\[
\left\|\big(\prod_{j\in \Omega} A^j-\prod_{j\in \Omega}M^j\big)(\rho)\right\|_1\leq \frac{\epsilon}{4}e^{\epsilon/4}.
\]
Since $e^{\epsilon/4}\leq e^{1/2}\leq 2$, we are done.
\end{proof}

\subsection{Approximating the output state}

Suppose that Eq.~\eqref{eq:suppose} holds and that $\rho$ is an $(n,k,D)$-state. We shall make use of the heavy slices whose existence is guaranteed by Lemma \ref{lem:lotsofslices} and the following matrix that is defined in terms of them:
\begin{equation}
\sigma=A^1A^2\ldots A^{T}(\rho).
\label{eq:sigmadef}
\end{equation}
Note that the operators $\{A^j\}$ depend on the location of the heavy slices and so we will have to compute this information (i.e., the guarantee that they exist is not enough); later we will account for the runtime required to do so. We now describe some properties of $\sigma$.

Noting from Eq.~\eqref{eq:Adef} that each $A^{j}$ is a real linear combination of CPTP maps, we see that $\sigma$ is Hermitian. For any Hermitian matrix $Q$ we have $\|Q\|_1 \geq \sum_{y} |\langle y|Q|y\rangle|$. Applying this along with Lemma \ref{lem:1norm} we get 
\begin{equation}
\frac{\epsilon}{2} \geq \left\|(A^1A^2\ldots A^{T}-M^1M^2\ldots M^{T})(\rho)\right\|_1 \geq \sum_{y\in \{0,1\}^n} |\langle y|\sigma|y\rangle-\langle y|\rho|y\rangle|,
\label{eq:eps2}
\end{equation}
where in the last inequality we used the fact that $M^{1}M^{2}\ldots M^{T}(\rho)$ has the same diagonal entries as $\rho$ since each $M^{j}$ is a CPTP map corresponding to measuring some of the qubits in the computational basis. Thus the matrix $\sigma$ satisfies the error guarantee Eq.~\eqref{eq:rhoerror} from Theorem \ref{thm:main} with $\epsilon\leftarrow \epsilon/2$.

From the definition Eq.~\eqref{eq:Adef} we see that $A^{j}$ is a linear combination of $L\equiv 2^{\ell}-1$ CPTP maps. Since $\ell=O(\log{n})$ we have $L=O(\mathrm{poly}(n))$. Moreover,  $\sigma$ is a pseudomixture
\begin{equation}
\sigma=\sum_{a_1,a_2,\ldots, a_T=1}^{L} f_{a_1}f_{a_2}\ldots f_{a_T} \cdot \mathcal{R}_{\vec{a}}(\rho),
\label{eq:pseudomixture}
\end{equation}
where $\{f_{a}\}_a$ are efficiently computable $\pm 1$-valued coefficients and each $\mathcal{R}_{\vec{a}}$ is a tensor product of CPTP maps,  each of which either measures a heavy slice in the computational basis, or traces out a heavy slice and replaces it with all-zeros. Crucially, due to the fact that $\omega\neq \emptyset$ in Eq.~\eqref{eq:Adef}, the CPTP map $\mathcal{R}_{\vec{a}}$ always traces out at least one heavy slice in each of the regions $H^{j}$ for $1\leq j\leq T$. Therefore we may write 
\begin{equation}
\mathcal{R}_{\vec{a}}(\rho)= \sigma^{1}_{a_1}\otimes \sigma^{2}_{a_1a_2}\otimes \sigma^{3}_{a_2a_3}\otimes \ldots \otimes \sigma^{T}_{a_{T-1}a_{T}}.
\label{eq:rprod}
\end{equation}
where $\sigma^1_{a_1}$ is a state whose support overlaps regions $V^{1}$ and $H^{1}$, $\sigma^2_{a_1a_2}$ is a state with support overlapping the three regions $H^{1}$, $V^{2}$ and $H^2$, etc. Note that the subsystems separated by the tensor product in Eq.~\eqref{eq:rprod} depend on $a_1,a_2,\ldots, a_T$.  For notational convenience, in the following we shall write $ \sigma^{1}_{a_0a_1}=\sigma^{1}_{a_1}$ for all $a_0$.

Each state $\sigma^{j}_{a_{j-1}a_j}$ can be prepared by starting with a reduced state of $|\psi\rangle$ on a rectangular region of width $O(\log{n})$ (an $(n,k-1,D)$-state), adjoining some fresh qubits in the all-zeros state, and then measuring some qubits. The latter operation does not alter the measurement statistics in the computational basis. In particular, 
the measurement statistics of $\sigma^{j}_{a_{j-1}a_j}$ are those of an $(n,k-1,D)$ state in tensor product with some ancilla qubits in the all-zeros state.

Finally, let us describe a method for computing ``marginals" of the form $\mathrm{Tr}(\sigma |x\rangle\langle x|_S\otimes I)$. In particular, we describe a polynomial-time reduction from this task (computing marginals of $\sigma$) to the task of computing marginals of states $\sigma^{1}_{a_1}, \sigma^{2}_{a_1a_2}, \ldots$ (this is nontrivial because the sum Eq.~\eqref{eq:pseudomixture} includes a superpolynomial number of terms).  As discussed above, this latter task reduces to computing marginals of $(n,k-1,D)$-states. 

 For any subset $S\subseteq [n]$ in the support of $\rho$, bitstring $x\in \{0,1\}^{|S|}$, and indices $a_1,a_2,\ldots, a_T\in [L]$ we can efficiently compute projectors $\Pi^1_{a_1}, \Pi^{2}_{a_1,a_2},\ldots , \Pi^{T}_{a_{T-1},a_T}$ diagonal in the computational basis, and each of the form $|y\rangle\langle y|\otimes I$ for a suitable substring $y$ of $x$, such that
\[
\mathrm{Tr}(\mathcal{R}_{\vec{a}}(\rho) |x\rangle\langle x|_S\otimes I)=\mathrm{Tr}(\sigma^{1}_{a_1}\Pi^{1}_{a_1})\mathrm{Tr}(\sigma^{2}_{a_1a_2}\Pi^{2}_{a_1a_2})\cdots \mathrm{Tr}(\sigma^{T}_{a_{T-1}a_T}\Pi^{T}_{a_{T-1}a_{T}}),
\]
and therefore
\[
\mathrm{Tr}(\sigma |x\rangle\langle x|_S\otimes I)=\sum_{a_1,a_2,\ldots, a_T=1}^{L}\left(f_{a_1}\mathrm{Tr}(\sigma^{1}_{a_1}\Pi^{1}_{a_1})\right)\left(f_{a_2}\mathrm{Tr}(\sigma^{2}_{a_1a_2}\Pi^{2}_{a_1a_2})\right)\cdots \big(f_{a_T}\mathrm{Tr}(\sigma^{T}_{a_{T-1}a_T}\Pi^{T}_{a_{T-1}a_{T}}))\big).
\]
The RHS is a matrix-vector product $v G^{2}G^{3}\ldots G^{T}w$ where $v$ is $1\times L$, $G^2, G^3,\ldots, G^{T}$ are $L\times L$ and $w$ is the $L\times 1$ all-ones vector. Each entry of each of these matrices and vectors is a marginal of one of the states $\sigma^{1}_{a_1}, \sigma^{2}_{a_1a_2}, \ldots$, and we need to compute all $O(TL^2)=O(\mathrm{poly}(n))$ such entries. Then we can compute the RHS using matrix-vector multiplication with runtime $O(TL^2)=O(\mathrm{poly}(n))$.

We are now ready to prove Theorem \ref{thm:main}. The proof uses the approximation $\sigma$ to the output state but proceeds slightly differently for $D=2$ and $D>2$ respectively.

\begin{proof}[Proof of Theorem \ref{thm:main}]

During the course of the algorithm described below we will need to compute certain sets of heavy slices whose existence is guaranteed by Lemma \ref{lem:lotsofslices} as long as Eq.~\eqref{eq:suppose} holds. If we run the algorithm on an input unitary $U$ that does not satisfy Eq.~\eqref{eq:suppose} then we may be unable to find them, in which case we output an error flag. Our algorithm also uses, as a subroutine,  the classical simulation algorithm for peaked constant-depth circuits from Section \ref{sec:peaked}; this algorithm may also output an error flag if Eq.~\eqref{eq:suppose} does not hold. For simplicity in the remainder of the proof we assume Eq.~\eqref{eq:suppose} holds.

First consider the case $D=2$. For two-dimensional circuits ($D=k=2$) the heavy slices from Lemma \ref{lem:lotsofslices} can be computed efficiently (as stated in the Lemma). In this case our approximation $K$ to the output state of the circuit is $K=\sigma$ defined in Eq.~\eqref{eq:sigmadef} with $\rho=|\psi\rangle\langle \psi|$. We have already (in Eq.~\eqref{eq:eps2}) established the error bound Eq.~\eqref{eq:rhoerror}. All that remains is to show that marginals of the form $\mathrm{Tr}(K |x\rangle\langle x|_S\otimes I)$ can be computed efficiently. Above we showed that this task reduces in $\mathrm{poly}(n)$ time to computing marginals of the states $\{\sigma^{j}_{a_{j-1}a_j}\}$. Moreover, this latter task is equivalent to computing marginals of $(n,1,2)$-states, i.e., a state prepared by a constant-depth circuit in a quasi-1D rectangular region of width $O(\log(n))$ and height $\sqrt{n}$. Such a circuit prepares a matrix product state of polynomial bond dimension, and marginals of such states can be computed efficiently using standard techniques.

Next suppose $D>2$. In this case we aim to establish Theorem \ref{thm:main} for $|\psi\rangle=U|0^n\rangle$ which is assumed to satisfy Eq.~\eqref{eq:suppose}. Here $U$ is a depth-$d$ circuit composed of geometrically local gates between nearest neighbor qubits on a $D$-dimensional grid such that each side has length $n^{1/D}$. By definition, the state $\psi$ of interest is an $(n,D,D)$-state. We will establish the theorem for $(n,k,D)$-states $\rho$. The proof is inductive in $k$ for $0\leq k\leq D$.

\textbf{Base case}
The base case of the induction is a $(n,0,D)$-state $\rho$. Let $J\subseteq [n]$ be the support of $\rho$, and let $J'$ be the lightcone of $J$ (with respect to the circuit $U$). From the definition of an $(n,0,D)$-state, we have $|J|=O(\log^D{n})$, and since $U$ has constant depth, we also have $|J'|=O(1)\cdot|J|=O(\log^D{n})$. Moreover,  
\[
\rho=\mathrm{Tr}_{[n]\setminus J}\left( U|0^{n}\rangle\langle 0^{n}|U^{\dagger}\right)=\mathrm{Tr}_{J'\setminus J}\left( U(J')|0^{|J'|}\rangle\langle 0^{|J'|}|U(J')^{\dagger}\right)
\]
where $U(J')$ is the subcircuit of $U$ consisting of all gates that act nontrivially on qubits in $J'$. Moreover, Eq.~\eqref{eq:suppose} implies
\begin{equation}
\mathrm{Tr}(\rho|h_J\rangle \langle h_J|)\geq n^{-a}
\label{eq:peakrho}
\end{equation}
where $h_J$ is the substring of $h$ consisting of bits in $J$. By Eq.~\eqref{eq:weightofh} the Hamming weight of $h_J$ is upper bounded as
\begin{equation}
w(h_J)\leq w(h)\leq c\log(n)
\label{eq:weightofhj}
\end{equation}
where $c$ is a known constant. The circuit $U(J')$ that prepares $\rho$ acts on $N=O(\log^D {n})$ qubits in total and has depth $d$. A naive classical simulation of this circuit would expend runtime $e^{O(\log ^D{n})}$ but we now show how to do better than this by exploiting Eqs.~(\ref{eq:peakrho},\ref{eq:weightofhj}). 

Let 
\[
\mathcal{C}=\{z\in \{0,1\}^{|J|}: w(z)\leq c\log(n)\}.
\]
From Eq.~\eqref{eq:weightofhj}, we have $h_J\in \mathcal{C}$. From Eqs.~(\ref{eq:peakrho},\ref{eq:weightofhj}) and Lemma \ref{lem:peakedshallowpurification} we can compute circuits $\{Q(z)\}_{z\in \mathcal{C}}$, each of which has depth $2d$ and acts on $m\leq 2N$ qubits, and a subset of qubits $S\subseteq [m]$, such that 
\be
\label{circuit_Q_peakness}
\mathrm{Tr}_{S}(Q(z)|0^m\rangle\langle 0^m|Q(z)^{\dagger})=\rho \quad \forall z\in \mathcal{C} \qquad \text{and} \qquad \max_{x\in \{0,1\}^m, z\in \mathcal{C} }|\langle x |Q(z)|0^m\rangle|^2\geq n^{-2a}.
\ee
The next step involves an exhaustive search over all the circuits $\{Q(z)\}_{z\in \mathcal{C}}$. The total number of circuits in this set is 
\[
|\mathcal{C}|\leq N^{c\log(n)}=n^{O(\log\log{n})}.
\] 
For each circuit $Q(z)$ we use the algorithm for simulation of peaked constant depth circuits from 
Theorem~\ref{thm:mainstate_generalized} of
Section \ref{sec:peaked} in a manner described in detail below.  The algorithm may output an error flag if it fails on the input circuit $Q(z)$. We terminate the exhaustive search over $z\in \mathcal{C}$ as soon as the algorithm succeeds. Eq.~\eqref{eq:peakrho} shows that there is at least one circuit in this set that is $2a$-peaked, which as we will see ensures that the algorithm succeeds for at least one $z\in \mathcal{C}$. The runtime of the algorithm for a given input unitary $Q(z)$ is upper bounded as $n^{O(\log\log{n})}$. So the total runtime is also upper bounded as $|\mathcal{C}|n^{O(\log\log{n})}=n^{O(\log\log{n})}$.

Let us now describe how we use the algorithm for simulation of peaked constant depth circuits from 
Theorem~\ref{thm:mainstate_generalized} of
Section \ref{sec:peaked} with an $m$-qubit circuit $Q=Q(z)$. Given $\epsilon=\Omega(1/\mathrm{poly}(n))$ satisfying $\epsilon\leq n^{-2a}$, if the algorithm succeeds it outputs a classical description of an $m$-qubit state $|\phi\rangle$ (described by a list of its nonzero coefficients in the computational basis) which satisfies $|\langle \phi|Q|0^m\rangle|^2\geq 1-\epsilon^2/4$. By Theorem~\ref{thm:mainstate_generalized}, the runtime of this algorithm 
scales as 
\[
T(m,\epsilon)
=m^{O(\log{m})+O(\log{(1/\epsilon)})}.
\]
Since $m=O(\log^D {n})$ with a constant $D$, the term
$O(\log{m})=O(\log\log{n})$ is negligible compared with $\log{(1/\epsilon)}=\Omega(\log{n})$. Therefore the algorithm has runtime
\[
T(m,\epsilon)=m^{O(\log{(1/\epsilon)})} =
n^{O(\log{\log{n}})}.
\]
The choice $\epsilon\leq n^{-2a}$ ensures the extra condition $\epsilon\le 3m^{1/2} P_{max}^{9/50}$ 
of Theorem~\ref{thm:mainstate_generalized}
is satisfied whenever $Q$ is $2a$-peaked (and thus $P_{max}\ge n^{-2a}$). As noted above, this ensures that the algorithm succeeds for at least one $z\in \mathcal{C}$.

The state $K=\mathrm{Tr}_S(|\phi\rangle\langle\phi|)$ then satisfies the conditions of Theorem \ref{thm:main}. In particular, 
\begin{align*}
\|K-\rho\|_1&=\|\mathrm{Tr}_S(|\phi\rangle\langle\phi|)-\mathrm{Tr}_S(Q|0^m\rangle\langle0^m| Q^{\dagger})\|_1\nonumber\\
&\leq \||\phi\rangle\langle \phi|-Q|0^m\rangle\langle0^m| Q^{\dagger}\|_1\\
&=2\sqrt{1-|\langle \phi|Q|0^m\rangle|^2}\\
&\leq \epsilon,
\end{align*}
and marginals of $K$ can be computed in time polynomial in the number of nonzero entries of $|\phi\rangle$ in the standard basis, i.e. $n^{O(\log\log{n})}$.

\textbf{Induction step}
Now suppose that Theorem \ref{thm:main} holds for $(n,k-1,D)$-states.  In the following we show this implies that it holds for $(n,k,D)$-states. So let $\rho$ be an $(n,k,D)$-state and let $\epsilon=\Omega(1/\mathrm{poly}(n))$ be given. Recall that Lemma \ref{lem:lotsofslices} guarantees the existence of heavy slices arranged in a certain regular pattern.  Moreover, whether or not a given rectangular region of width $O(d)$ is a heavy slice is determined by the marginal probability of measuring all-zeros in an $(n,k-1,D)$-state. Thus, by our inductive hypothesis, there exists an algorithm with runtime $n^{O(\log\log{n})}$ to check whether or not a given rectangular region of width $O(d)$ is a heavy slice. In this way, with runtime $n^{O(\log\log{n})}$ we can compute the heavy slices from Lemma \ref{lem:lotsofslices}.

Once we have computed the heavy slices, we use them to specify our Hermitian matrix $\sigma$ from Eq.~\eqref{eq:sigmadef}. We have already shown in Eq.~\eqref{eq:eps2} that 
\begin{equation}
\| \mathcal{M}( \sigma)- \mathcal{M}(\rho)\|_1\leq\epsilon/2.
\label{eq:msigma1}
\end{equation}
Below we show, using our inductive hypothesis, that 
\begin{equation}
\|\mathcal{M}(K)-\mathcal{M}(\sigma)\|_1\leq \epsilon/2,
\label{eq:msigma2}
\end{equation}
for some Hermitian matrix $K$ such that marginals $\mathrm{Tr}(K |x\rangle\langle x|_S)$ can be computed with runtime $n^{O(\log\log{n})}$. Combining Eqs.~(\ref{eq:msigma1},\ref{eq:msigma2}) using the triangle inequality gives the desired error bound Eq.~\eqref{eq:rhoerror}.

To construct $K$ from $\sigma$, recall that each of the states $\sigma^{j}_{a_{j-1}a_j}$ in the tensor product Eq.~\eqref{eq:rprod} has the same measurement statistics as an $(n,k-1,D)$-state, in tensor product with some fresh qubits in the all-zeros state. That is,
\begin{equation}
\mathcal{M}(\sigma^{j}_{a_{j-1}a_j})=\mathcal{M}(\tilde{\sigma}^{j}_{a_{j-1}a_j})\otimes \mathcal{M}(|0\rangle\langle 0|_S).
\label{eq:msigma}
\end{equation}
for some subset of qubits $S$, and where $\tilde{\sigma}^{j}_{a_{j-1}a_j}$ is an $(n,k-1,D)$-state.

Starting from the $(n,k-1,D)$-state $\tilde{\sigma}^{j}_{a_{j-1}a_j}$ and a precision parameter $\delta=\Omega(1/\mathrm{poly}(n))$, we can use the inductive hypothesis to obtain a Hermitian matrix $\tilde{K}^j_{a_{j-1}a_j}$, that approximates $\tilde{\sigma}^{j}_{a_{j-1}a_j}$ in the sense of Theorem \ref{thm:main}. Letting $K^{j}_{a_{j-1}a_j}=\tilde{K}^j_{a_{j-1}a_j}\otimes |0\rangle\langle 0|_S$ and using Eq.~\eqref{eq:msigma},  we see that 
\begin{equation}
\|\mathcal{M}(\sigma^{j}_{a_{j-1}a_j})-\mathcal{M}(K^{j}_{a_{j-1}a_j})\|_1\leq \delta 
\label{eq:msigma3}
\end{equation}
and that there exists a classical algorithm with runtime  $n^{O(\log \log{n})}$ to exactly compute``marginals" $\mathrm{Tr}(K^{j}_{a_{j-1}a_j} |x\rangle \langle x|_S\otimes I)$ for any 
 $S\subseteq [n]$ and $x\in \{0,1\}^{|S|}$. Our matrix $K$ is obtained by replacing $\sigma^{j}_{a_{j-1}a_j}$ in Eqs.~(\ref{eq:pseudomixture}, \ref{eq:rprod}) with $K^{j}_{a_{j-1}a_j}$, i.e., 
\begin{equation}
K=\sum_{\vec{a}}\bigotimes_{j=1}^T f_{a_j}K_{a_{j-1}a_j}^j.
\label{eq:rhodef}
\end{equation}

We have already shown above that computing marginals of such a state $K$ can be reduced in $\mathrm{poly}(n)$-time (via matrix multiplication) to computing marginals of the individual states $K_{a_{j-1}a_j}^j$, which takes runtime $n^{O(\log \log{n})}$. To complete the proof it remains to establish Eq.~\eqref{eq:msigma2}, which is obtained from the following Lemma by setting $\delta=(8TL^2)^{-1}\epsilon$.
\begin{lemma}
If Eq.~\eqref{eq:msigma3} holds for every $j\in[T]$ and $a_j\in[L]$ for some $\delta\leq\frac{1}{2TL^2}$, then $\lVert\mathcal{M}(K)-\mathcal{M}(\sigma)\|_1\leq 4T L^2\delta$.
\label{lem:robust}
\end{lemma}
\begin{proof}
For reference, we restate the definition of $\sigma$ from Eqs. \eqref{eq:sigmadef}, \eqref{eq:pseudomixture}, \eqref{eq:rprod}, 
\[
\sigma=\prod_{j\in[T]}A^j(\rho)=\sum_{a_i=1:i\in[T]}^L\bigotimes_{j\in[T]} f_{a_j}\sigma_{a_{j-1}a_j}^j
\]
where $A^1,\ldots,A^T$ are defined in Eq. \eqref{eq:Adef}. Let $\Delta_{a_{j-1}a_j}^j=K_{a_{j-1}a_j}^{j}-\sigma_{a_{j-1}a_j}^{j}$, so that Eq.~\eqref{eq:msigma3} becomes 
\begin{equation}
\|\mathcal{M}(\Delta_{a_{j-1}a_j}^j)\|_1\leq \delta.
\label{eq:mofr}
\end{equation}
For every $S\subseteq [T]\setminus\{\emptyset\}$, define
\begin{equation}
\overline{S}=\bigg(\bigcup_{j\in S}\{j-1,j\}\bigg)\setminus\{0\}.
\end{equation}
We can write
\begin{align}
&K-\sigma=\sum_{\vec{a}}\bigotimes_{j=1}^T f_{a_j}(\sigma_{a_{j-1}a_j}^j+\Delta_{a_{j-1}a_j}^j)-\sum_{\vec{a}}\bigotimes_{j=1}^T f_{a_j}\sigma_{a_{j-1}a_j}^j\\
&=\sum_{\vec{a}}\sum_{S\subseteq [T]\setminus\{\emptyset\}}\Big(\bigotimes_{j\in S}f_{a_j}\Delta_{a_{j-1}a_j}^j\Big)\otimes\Big(\bigotimes_{j\notin S}f_{a_j}\sigma_{a_{j-1}a_j}^j\Big)\\
&=\sum_{S\subseteq[T]\setminus\{\emptyset\}}\sum_{a_i=1:i\in \overline{S}}^L\bigg(\prod_{j\in\overline{S}\setminus S}f_{a_j}\bigotimes_{j\in S}f_{a_j}\Delta_{a_{j-1}a_j}^j\bigg)\otimes\bigg(\sum_{a_i=1:i\notin \overline{S}}^L\Big(\bigotimes_{j\in\overline{S}\setminus S}\sigma_{a_{j-1}a_j}^j\Big)\otimes \Big(\bigotimes_{j\notin\overline{S}}f_{a_j}\sigma_{a_{j-1}a_j}^j\Big)\bigg).\nonumber
\end{align}
Let $S\subseteq [T]\setminus\{\emptyset\}$. Suppose we fix the value of $a_j\in[L]$ for every $j\in\overline{S}$ and consider
\begin{equation}
\sum_{a_i=1:i\notin \overline{S}}^L\Big(\bigotimes_{j\in\overline{S}\setminus S}\sigma_{a_{j-1}a_j}^j\Big)\otimes \Big(\bigotimes_{j\notin\overline{S}}f_{a_j}\sigma_{a_{j-1}a_j}^j\Big).
\label{eq:tilde_rho}
\end{equation}
Recall that for every $j\in[T]$, the index $a_j\in[L]$ specifies a subset $\omega\subseteq[\ell]$ corresponding to one term $\prod_{i\in [\ell]\setminus\omega} M_{s^{j}_{i}} \prod_{i\in \omega} E_{s^{j}_{i}}$ in the definition of $A^j$ (see Eq. \eqref{eq:Adef}). For every $j\in\overline{S}$, let $\mathcal{R}^j_{a_j}$ denote the term in the definition of $A^j$ specified by $a_j$, so that 
\[
\mathcal{R}_{\vec{a}}=\mathcal{R}^1_{a_1}\otimes \mathcal{R}^2_{a_2}\ldots \otimes \mathcal{R}^T_{a_T},
\]
where $\mathcal{R}_{\vec{a}}$ is the operator from Eqs.~(\ref{eq:pseudomixture},\ref{eq:rprod}).

 By definition we have
\begin{equation}
\bigotimes_{j\in\overline{S}}\mathcal{R}_{a_j}^j(\rho)=\Big(\bigotimes_{j\in S}\sigma_{a_{j-1}a_j}^j\Big)\otimes\tilde{\rho}
\end{equation}
for some density matrix $\tilde{\rho}$. For every $j\in S$, let $\tr_{a_{j-1}a_j}^j$ denote the operation of tracing out all the qubits in $\sigma_{a_{j-1}a_j}^j$, and note that

\begin{equation}
\tilde{\rho}=\mathcal{E}(\rho) \qquad \mathcal{E}\equiv \prod_{j\in S}\tr_{a_{j-1}a_j}^j\bigg(\bigotimes_{j\in\overline{S}}\mathcal{R}^j_{a_j}\bigg).
\end{equation}
and
\begin{equation}
\prod_{j\in[T]\setminus\overline{S}}A^j(\tilde{\rho})=\sum_{a_i=1:i\in[T]\setminus \overline{S}}^L\bigg(\bigotimes_{j\in\overline{S}\setminus S}\sigma_{a_{j-1}a_j}^j\bigg)\otimes \bigg(\bigotimes_{j\in[T]\setminus\overline{S}}f_{a_j}\sigma_{a_{j-1}a_j}^j\bigg)
\end{equation}
which is precisely Eq. \eqref{eq:tilde_rho}. Now
\begin{align}
\bigg\|\prod_{j\in[T]\setminus\overline{S}}A^j(\tilde{\rho})-\prod_{j\in[T]\setminus\overline{S}}M^j(\tilde{\rho})\bigg\|_1&= \bigg\|\mathcal{E}\bigg(\prod_{j\in[T]\setminus\overline{S}}A^j(\rho)-\prod_{j\in[T]\setminus\overline{S}}M^j(\rho)\bigg)\bigg\|_1\\
&\leq \bigg\|\prod_{j\in[T]\setminus\overline{S}}A^j(\rho)-\prod_{j\in[T]\setminus\overline{S}}M^j(\rho)\bigg\|_1
\label{eq:ajmj}
\end{align}
where in the second line we used the fact that $\|\mathcal{E}(Q)\|_1\leq \|Q\|_1$ for all Hermitian matrices $Q$, since $\mathcal{E}$ is a CPTP map.
Now we can use Lemma \ref{lem:1norm} to upper bound the RHS of Eq.~\eqref{eq:ajmj} by $\epsilon/2$. Then, by the reverse triangle inequality,
\begin{equation}
\bigg\| \prod_{j\notin\overline{S}}A^j(\tilde{\rho})\bigg\|_1\leq\bigg\|\prod_{j\notin\overline{S}}M^j(\tilde{\rho})\bigg\|_1+\epsilon/2=1+\epsilon/2\leq 2.
\label{eq:reverset}
\end{equation}
Therefore,
\begin{align}
\lVert \mathcal{M}(K)-\mathcal{M}(\sigma)\rVert_1&\leq \sum_{S\subseteq[T]\setminus\{\emptyset\}}\sum_{a_i=1:i\in \overline{S}}^L\bigg\lVert \mathcal{M}\bigg(\prod_{j\in\overline{S}\setminus S}f_{a_j}\bigotimes_{j\in S}f_{a_j}\Delta_{a_{j-1}a_j}^j\bigg)\otimes\mathcal{M}\bigg(\prod_{j\notin\overline{S}}A^j(\tilde{\rho})\bigg)\bigg\rVert_1\nonumber\\
&\leq \sum_{S\subseteq[T]\setminus\{\emptyset\}}L^{|\overline{S}|} \delta^{|S|}2,\label{eq:almostdone}
\end{align}
where we used Eqs.~(\ref{eq:mofr}, \ref{eq:reverset}) and the fact that $\|\mathcal{M}(Q)\|_1\leq \|Q\|_1$ for any Hermitian matrix $Q$.  Continuing from Eq.~\eqref{eq:almostdone} we arrive at
\begin{align*}
\lVert \mathcal{M}(K)-\mathcal{M}(\sigma)\rVert_1 &\leq 2\sum_{k=1}^T \binom{T}{k}L^{2k}\delta^k\\
&= 2\left((1+L^2\delta)^T-1\right)\\
&\leq 2T L^2\delta e^{T L^2\delta}\\
&\leq 4TL^2\delta.
\end{align*}
where in the last line we used $\delta\leq (2TL^2)^{-1}$ to infer $e^{T L^2\delta}\leq e^{1/2}\leq 2$.
\end{proof}
\end{proof}

\section{Conclusions}
\label{sec:conclusions}
We have identified a large family of $n$-qubit quantum circuits which can be simulated on a classical computer in quasipolynomial time
scaling as $n^{O(\log{n})}$.
These are peaked shallow circuits which have depth $O(1)$ and whose output distribution assigns the probability $n^{-O(1)}$
to at least one bit string. 
To this end we showed that the output distribution of a peaked shallow circuit has almost all its mass in a Hamming ball of diameter $O(\log n)$.

We also drastically simplified and extended the scope of simulation algorithms based on heavy slices and 
inclusion-exclusion principle
pioneered in Refs.~\cite{coudron, coudron2}. This results in a more favourable simulation runtime which scales as
$n^{O(1)}$ or  $n^{O(\log{\log{n}})}$
if the circuit is geometrically local on a $D$-dimensional grid of qubits with $D=2$ or $D\ge 3$ respectively.

Our work leaves several open questions. (1) Is there a polynomial time classical algorithm for simulating peaked shallow circuits?
(2) Can we approximate the output state $|\psi\rangle$ of a peaked shallow circuit including the overall phase ?
Note that our algorithms only aim to approximate the projector $|\psi\rangle\langle \psi|$ onto the output state. As a consequence, our algorithms are not capable of recovering the sign $\pm$ of expected values of 
Pauli observables which are of great interest for variational quantum algorithms.
(3) What is complexity of simulating deeper peaked circuits, say with the depth $O(\log{n})$ ?
Our analysis breaks down for such circuits since the lightcone of a single output qubit may span all $n$ input qubits.
Finally, we conjecture that our classical simulation algorithm 
for general $n$-qubit peaked shallow circuits can be 
implemented by a quantum circuit of size $\mathrm{poly}(n)$ acting on $\mathrm{poly}(\log{n})$ qubits.
For example, one may be able to prepare the largest eigenvector $|\phi\ra$ of the projected Hamiltonian $G$ defined in Section~\ref{sec:peaked}
using Quantum Phase Estimation.
The required number of qubits is only poly-logarithmic in $n$ since the projected Hamiltonian $G$ acts on the Hilbert space of dimension
$n^{O(\log{n})}$, see Section~\ref{sec:peaked}.
If true, this conjecture would provide an interesting example of quantum circuit knitting~\cite{bravyi2016trading,peng2020simulating,piveteau2023circuit}
where a large quantum circuit is transformed into a form that can be simulated on a 
small quantum device.

\section*{Acknowledgments}
SB thanks Natalie Parham for helpful discussions. 
DG and YL acknowledge the support of IBM Research, as well as the
Natural Sciences and Engineering Research Council of
Canada through grant number RGPIN-2019-04198. DG is a fellow of the Canadian Institute
for Advanced Research, in the quantum information science program.
Research at
Perimeter Institute is supported in part by the Government of Canada through the Department of Innovation,
Science and Economic Development Canada and by the
Province of Ontario through the Ministry of Colleges and
Universities.
\bibliographystyle{unsrt}
\bibliography{bibliog}
\end{document}